\documentclass[preprint,11pt]{elsarticle}

\usepackage[utf8]{inputenc}
\usepackage{amssymb}
\usepackage{amsmath}
\usepackage{amsthm}
\usepackage{amsfonts}
\usepackage[table,svgnames]{xcolor}
\usepackage{array}
\usepackage{ifthen}
\usepackage{stmaryrd}
\usepackage{pgf,tikz}
\usepackage{mathrsfs}
\usepackage{multirow}
\usepackage{multicol}
\usepackage{makecell}
\usepackage{placeins}
\usepackage{diagbox}

\usepackage[colorlinks=true,
            linkcolor=blue]
           {hyperref}

\journal{Theoretical Computer Science}

\newcommand{\N}{\mathbb N}

\newcommand{\Z}{\mathbb Z}
\renewcommand{\O}{\mathscr{O}} 
\newcommand{\C}{\mathscr{C}} 

\newcommand{\B}{\mathbb{B}}
\newcommand{\Bn}{\mathbb{B}^n}
\newcommand{\1}{\mathtt{1}}

\newcommand{\ha}{\hat{a}}
\newcounter{finaln}
\newcommand{\integers}[2][2]{\ifthenelse{\equal{#1}{2}}
{\llbracket #2\rrbracket}
{\ifthenelse{\equal{#1}{0}}{
\ifthenelse{\equal{11}{\the\catcode`#2}}
{\{0,\ldots,#2 -1\}}
{\newcounter{finaln}
\setcounter{finaln}{#2 -1} 
\{0,\ldots,\thefinaln \}}}
{\{1,\ldots,#2\}}}}

\newtheorem{thm}{Theorem}
\newtheorem{corollary}{Corollary}

\newtheorem{lemma}[thm]{Lemma}

\newtheorem{conjecture}{Conjecture}

\newtheorem{definition}{Definition}
\newtheorem{remark}{Remark}

\newcommand{\ie}{i.e.{}}
\newcommand{\eg}{e.g.{}}

\begin{document}


\begin{frontmatter}


\author[label1,label2,label4]{Isabel Donoso-Leiva}
\author[label1]{Eric Goles}
\author[label1]{Martín Ríos-Wilson}
\author[label2,label3]{Sylvain Sené}

\affiliation[label1]{
    organization={Facultad de Ingeniería y Ciencias, Universidad Adolfo Ibáñez},
    city={Santiago},
    country={Chile}
}
\affiliation[label2]{
    organization={Aix Marseille Univ, CNRS, LIS},
    city={Marseille},
    country={France}
}
\affiliation[label3]{
    organization={Université publique},
    city={Marseille},
    country={France}
}
\affiliation[label4]{
    organization={Facultad de Ingeniería, Universidad Católica de Temuco},
    city={Temuco},
    country={Chile}
}

\title{On the Convergence of Elementary Cellular Automata under Sequential Update Modes}

\begin{abstract}
In this paper, we perform a theoretical analysis of the sequential convergence of elementary cellular automata that have at least one fixed point. Our aim is to establish which elementary rules always reached fixed points under some sequential update modes, regardless of the initial configuration. In this context, we classify these rules according to whether all initial configurations converge under all, some, one or none sequential update modes, depending on if they have fixed points under synchronous (or parallel) update modes. 
\end{abstract}

%



\end{frontmatter}

\section{Introduction}
\label{sec:intro}
Cellular automata are simple systems that can hold great complexity~\cite{Wolfram1984}. They are composed of cells ordered in a grid that interact with one another following a local rule over discrete time. They were first introduced by Ulam and von Neuman~\cite{vonNeumann1966}, which were inspired in the work on neural networks done by McCulloch and Pitts~\cite{McCulloch1943}, {with Kleene~\cite{Kleene1956} formalizing the definition of finite automata in 1956.}
Cellular Automata have since proved to be useful in multiple areas of science, such as biology and health sciences~\cite{White2007, Sirakoulis2000}, social models~\cite{Sakoda1971,Schelling1971}, parallel and distributed computation~\cite{Smith1972,Mazoyer1987}, and physics~\cite{Vichniac1984,Boghosian1999,Salo2017}.

In the original definition of the cellular automata model, all cells are updated simultaneously. However, issues with synchronicity have been raised from both mathematical modeling and computational perspectives~\cite{fates2013guided}. Therefore, \textit{asynchronous update modes} have been proposed such as: fully asynchronous~\cite{fates2006fully}, $\alpha$-asynchronous~\cite{Fates2005}, random sweep~\cite{SCHONFISCH1999,kitagawa1974cell} as well as others which are probabilistic in nature. 

In~\cite{Robert1986, goles2021complexity, goles2016pspace, rios2024c} \textit{deterministic update modes} are studied which have the advantage of breaking the dependency on parallelism without introducing randomness.  These update modes have opened interesting areas of study, such as:
the sensitivity to synchronism in elementary cellular automata with respect to the number of possible resulting dynamics in~\cite{perrot2020,balbi2022} 
and with respect to the longest limit cycle in~\cite{Donoso2025}, the complexity of the reachability problem is studied in~\cite{barrett2006}, while in the context of Boolean networks (of which elementary cellular automata are particular cases), their convergence under non-deterministic and deterministic update modes are studied in~\cite{melliti2013}, the robustness of their dynamical properties under different deterministic update modes are established in~\cite{Aracena2009}, different kinds of non-deterministic and deterministic Boolean networks are compared with regards to their attractors in~\cite{gershenson2004} and in~\cite{Perrotin2023} limit cycles of automata networks are studied.

The uses of these update modes are varied, for example: 
in~\cite{MARGOLUS1984} block-sequential cellular automata are proposed as representations of physical concepts, 
in~\cite{Pauleve2022} asynchronous deterministic Boolean networks are used to portray biological systems,
in~\cite{perrotin2024,ruivo2024} sequential cellular automata are proposed as a solution to the density classification task and in~\cite{cornforth2005} automata networks under ordered asynchronous update modes are suggested to depict different artificial and biological systems. {In~\cite{das2013} necessary and sufficient conditions are established for an ECA to converge under asynchronous updating, with the complete proof being provided in~\cite{Sethi2016}. In~\cite{Sethi2018size} additional conditions over the size of the ring are studied to ensure the convergence of other ECAs under asynchronous updating. In~\cite{Sethi2014,Fates2018,Roy2024} reversibility of ECAs under stochastic updating is studied.}

We have chosen to utilize deterministic update modes, among them, we have focused on sequential update modes, which consist on individual cells being updated one after another according to a previously established order. {We emphasize that sequential update modes are a particular kind of asynchronous update modes, and as such they can have distinct properties. In addition to being deterministic, they allow exactly one node to update its state per iteration, and any given node cannot be updated more than once before all the nodes have been updated, according to a specific order given by a permutation of the set of nodes. The use of sequential update modes allows us to exploit combinatorial tools~\cite{Sawada2003}, since the dynamics are defined over a ring, several sequential update orders lead to equivalent dynamics, which enables us to disregard redundant update modes. This, in turn, allows us to use methods from discrete mathematics~\cite{Aracena2011} and to focus our analysis on non-equivalent update modes only.}

\subsection{Our Contribution}
In~\cite{Robert1986} it was proven that elementary cellular automata under block-sequential update modes (which include the sequential ones) have exactly the same fixed points as they do under parallel update mode. Therefore, we study the convergence of elementary cellular automata that have fixed points in parallel, knowing that ``new'' fixed points will not appear under sequential update modes.

Thus, in this paper we classify elementary rules according to {whether} there exists a sequential update mode that leads all configurations to a fixed point, from where we introduce the concepts of \textit{universal} and \textit{covering of} update modes.
For a given rule, a covering is a set of update modes such that for any initial configuration, there exists an element in the set, that if the rule is applied following the order established by it, the configuration reaches a fixed point. If a set containing only one update mode is a covering, then that update mode is universal.

We start with the rules characterized in~\cite{das2013} as always reaching fixed points under asynchronous updating, proving that sequential update modes in particular are enough to replicate the result for all but two of these rules.

For the 50 {non-equivalent} rules that fulfill the conditions of established in~\cite{das2013}, we provide detailed proofs for rules numbered 2, 26, and 104 according to Wolfram's code~\cite{Wolfram1983}, because the other rules follow similar reasoning as the ones we have chosen to showcase. 
Notably, we found that for Rule $106$, which is a complex rule according to Wolfram's classification~\cite{Wolfram1984}, there is a sequential update mode that leads all configurations to a fixed point, but only when the size of the ring is even. Here, we introduce the notion of $X$-universal, meaning that a given rule under that update mode reaches a fixed point from all configurations belonging to a subset $X$. 
A summary of this part of the classification is provided in Table~\ref{tab:inside-theorem}.

We continue our study with the 38 non-equivalent elementary cellular automata that fall outside the purview of said paper and show that five of them can always reach fixed points, when said fixed points exist. This means that the result is weaker than the one obtained for the previous rules, since fixed points do not necessarily exist for all ring sizes, but when they do exist, there are sequential update modes that lead to them. {Some of the rules with which we work on this section have been studied in~\cite{Sethi2018size} under asynchronous updating. Here, we show that sequential update modes are enough to ensure the convergence of some of them, exhibiting a specific update mode that ensures the convergence of all configurations (when the size of the ring is even) for Rules $7$ and $15$. On the other hand, we show that not all configurations can converge under sequential update modes for Rules $6$, $14$ and $22$.}
In addition, we prove that there are rules{, a few of which have been studied for the asynchronous case in~\cite{Fates2018},} which have fixed points that cannot be reached under any sequential update mode from certain configurations. The outline of the rules of this section is provided in Table~\ref{tab:outside-theorem}. 

\subsection{Structure of the Paper}
In Section~\ref{sec:def-and-notations}, the main definitions and notations are formalized. Section~\ref{sec:results} contains the demonstrations that justify the classification that is being proposed. In the first part, we will prove that most of the 50 rules given by the theorem in~\cite{das2013} can always reach fixed points under sequential update modes; the second part focuses on elementary cellular automata outside the theorem, where additional conditions are needed to ensure the convergence to fixed points. The paper ends with Section~\ref{sec:discussion}, in which we discuss the results as well as some future work.

\section{Preliminaries}
\label{sec:def-and-notations}
\subsection{Elementary Cellular Automata and Sequential Update Modes}

{Let $\integers{n} = \integers[0]{n}$,
let $\B = \{0,1\}$, and let $x_i$ denote the $i$-th component of vector $x \in \B^n$.
Given a vector $x \in \B^n$, we can denote it classically as 
$(x_0, \dots, x_{n-1})$ or as the word $x_0 \dots x_{n-1}$ if it eases the reading.}

In broad terms, a \emph{cellular automaton} (CA) is a collection of cells, each of which will have a state within the alphabet $Q$. The cells interact with each other over discrete time according to a rule which is defined by the state of the cells in their respective neighborhood. 
Formally, a CA is a tuple $(\Z^d,Q,N,F)$ composed of
\begin{itemize}
\item $\Z^d$ the $d-$dimensional cellular space, also known as \emph{grid}, which can be finite or infinite,
\item $Q$ the finite set of states that a cell can have, called the \emph{alphabet},
\item $N$ the neighborhood, which associates a cell of the grid with its neighbors,
\item $F$ the rule, which is defined by local functions $f_i^{N_i}\to Q$, where $f_i$ is the $i-$th component of $F$ and $N_i$ is the neighborhood of cell $i$.
\end{itemize}

An \emph{elementary cellular automaton} (ECA) is a kind of CA such that the grid is one-dimensional, the alphabet is $\B$, the local function $f_i$ is the same for the entire grid and it depends on the state of the cells $i-1$, $i$ itself, and $i+1$. From where it follows that there are $2^{2^{3}}=256$ distinct ECAs, of which there are $88$ that are not equivalent through symmetry~\cite{li1990structure}. {Since in this case the local function is the same for the entire grid, it is denoted simply by $f$, or $f_w$ if clarification is needed, where $w$ is the corresponding Wolfram code~\cite{Wolfram1983}.}

As previously stated, the cellular space, $Z^d$, can be either infinite or finite. In this paper we have chosen to work with finite ECAs, which effectively means that the grid is viewed as a {one-dimensional} torus, that is, a ring of integers modulo $n$: $\Z/n\Z$. As such, the neighborhood of cell $0$ is $\{n-1,0,1\}$ and that of cell $n-1$ is $\{n-2,n-1,0\}$. Thus, a configuration $x$ is an element of $\B^n$, is formally defined by a Boolean vector of dimension $n$ and can equivalently be represented by a binary word in what follows for the sake of clarity and concision, according to the context (\eg~$x=(0,1,1,0,1)\equiv01101$).\smallskip

We can think of an update mode as a function that associates each iteration to the set of cells that will be updated in it, that is  $\mu^{\star}:\N\to\mathcal{P}(\integers{n})$, {with $\mathcal{P}(\integers{n})$ the power set of $\integers{n}$}. In this paper we will work with \textit{periodic update modes}, which means that there exists $p\in\N$ such that for all $t\in\N$, $\mu^{\star}(t+p)=\mu^{\star}(t)$. In other words an update mode of period $p$ can be viewed as a finite sequence of $p$ subsets of cells. Each iteration in which a subset of cells is updated is called a \textit{substep}, while a \textit{step} is the composition of the  substeps having occurred over a period $p$. 

\begin{definition}[Sequential update modes]
\emph{Sequential update modes} are a family of periodic update modes which are defined as $\mu = (\phi(\integers{n}))$, where $\phi(\integers{n}) = \{i_0\}, \dots, \{i_{n-1}\}$ is a permutation of $\integers{n}$ {that establishes the order of priority with which the cells are updating, such that} one and only one cell update its state at each substep so that all cells have updated their state after $n$ substeps (\ie, after each step) depending on the order induced by $\phi$. Naturally, for these update modes the period is always equal to the number of cells.
\end{definition}

Note that while the formal notation for a sequential update mode is $\mu=(\{i_0\},\{i_1\},\dots,\{i_{n-1}\})$, namely a sequence of $n$ subsets of cardinality $1$, we make an abuse of notation in what follows by using just a sequence of elements: $\mu=(i_0,i_1,\dots,i_{n-1})$ to facilitate reading.

An \emph{asynchronous update mode} is an update mode in which not all cells are updated simultaneously. Naturally, sequential update modes fall under this definition. However, asynchronous update modes can also be non-deterministic, for instance if {the order of priority with which we update} the cells is selected at random.

\begin{figure}[t]
	\centerline{
	\begin{minipage}{.3\textwidth}
	\centerline{
	\scalebox{1}{\begin{tikzpicture}[help lines/.style={draw=gray!50},
		every node/.style={help lines,rectangle,minimum size=3mm},
		ca/.style={draw=none,row sep=0mm,column sep=0mm, ampersand 
		replacement=\&},
		st0/.style={fill=white,help lines},
		stg/.style={fill=gray,help lines},
		st1/.style={fill=black,help lines}]

		\matrix[ca] {
		\node[st0] {};\& \node[st1] {};\& \node[st0] {};\& \node[st0] {};\\
		\node[stg] {};\& \node[stg] {};\& \node[st0] {};\& \node[st0] {};\\
        \node[stg] {};\& \node[stg] {};\& \node[st0] {};\& \node[st0] {};\\
		\node[stg] {};\& \node[stg] {};\& \node[stg] {};\& \node[st0] {};\\
		\node[st1] {};\& \node[st1] {};\& \node[st1] {};\& \node[st0] {};\\
		\node[stg] {};\& \node[stg] {};\& \node[stg] {};\& \node[st0] {};\\
        \node[stg] {};\& \node[stg] {};\& \node[st0] {};\& \node[st0] {};\\
		\node[stg] {};\& \node[st0] {};\& \node[stg] {};\& \node[st0] {};\\
		\node[st1] {};\& \node[st0] {};\& \node[st0] {};\& \node[st0] {};\\
		};
	\end{tikzpicture}}}
	\end{minipage}
	\quad
	\begin{minipage}{.33\textwidth}
	\centerline{
	\scalebox{1}{\begin{tikzpicture}[help lines/.style={draw=gray!50},
		every node/.style={help lines,rectangle,minimum size=3mm},
		ca/.style={draw=none,row sep=0mm,column sep=0mm, ampersand 
		replacement=\&},
		st0/.style={fill=white,help lines},
		stg/.style={fill=gray,help lines},
		st1/.style={fill=black,help lines}]
        
		\matrix[ca] {
		\node[st0] {};\& \node[st1] {};\& \node[st0] {};\& \node[st0] {};\\
		\node[st0] {};\& \node[stg] {};\& \node[stg] {};\& \node[st0] {};\\
        \node[stg] {};\& \node[stg] {};\& \node[stg] {};\& \node[st0] {};\\
		\node[stg] {};\& \node[stg] {};\& \node[stg] {};\& \node[st0] {};\\
		\node[st1] {};\& \node[st0] {};\& \node[st1] {};\& \node[st0] {};\\
		\node[stg] {};\& \node[st0] {};\& \node[st0] {};\& \node[st0] {};\\
        \node[st0] {};\& \node[st0] {};\& \node[st0] {};\& \node[st0] {};\\
		\node[st0] {};\& \node[st0] {};\& \node[st0] {};\& \node[st0] {};\\
		\node[st0] {};\& \node[st0] {};\& \node[st0] {};\& \node[st0] {};\\
		};
	\end{tikzpicture}}}
	\end{minipage}
	\quad
	\begin{minipage}{.3\textwidth}
	\centerline{
	\scalebox{1}{\begin{tikzpicture}[help lines/.style={draw=gray!50},
		every node/.style={help lines,rectangle,minimum size=3mm},
		ca/.style={draw=none,row sep=0mm,column sep=0mm, ampersand 
		replacement=\&},
		st0/.style={fill=white,help lines},
		stg/.style={fill=gray,help lines},
		st1/.style={fill=black,help lines}]
		\matrix[ca] {
		\node[st0] {};\& \node[st1] {};\& \node[st0] {};\& \node[st0] {};\\
		\node[st0] {};\& \node[stg] {};\& \node[st0] {};\& \node[st0] {};\\
        \node[st0] {};\& \node[stg] {};\& \node[stg] {};\& \node[st0] {};\\
		\node[st0] {};\& \node[stg] {};\& \node[stg] {};\& \node[st0] {};\\
		\node[st1] {};\& \node[st1] {};\& \node[st1] {};\& \node[st0] {};\\
		\node[stg] {};\& \node[stg] {};\& \node[stg] {};\& \node[st0] {};\\
        \node[stg] {};\& \node[stg] {};\& \node[stg] {};\& \node[st0] {};\\
		\node[stg] {};\& \node[st0] {};\& \node[stg] {};\& \node[st0] {};\\
		\node[st0] {};\& \node[st0] {};\& \node[st1] {};\& \node[st0] {};\\
		};
	\end{tikzpicture}}}
	\end{minipage}}
	
	\centerline{
	\begin{minipage}{.33\textwidth}
	\centerline{(a)}
	\end{minipage}
	\quad
	\begin{minipage}{.33\textwidth}
	\centerline{(b)}
	\end{minipage}
	\quad
	\begin{minipage}{.3\textwidth}
	\centerline{(c)}
	\end{minipage}
	}
	\caption{Space-time diagrams (time going downward) of configuration 
		$0100$ following Rule $45$ under: 
		(a) the update mode $\mu_{1} = (0,1,2,3)$, 
		(b) the update mode $\mu_{2} = (2,0,3,1)$, 
		(c) the update mode $\mu_{3} = (3,2,1,0)$. Configurations colored in gray (resp. black) represent the ones obtained at substeps (resp. steps).}
	\label{fig:space-time}
\end{figure}

\subsection{Dynamical Systems}

An ECA $f$ along with an update mode $\mu$ defines a discrete dynamical system on $\Bn\rightarrow\Bn$. More formally in our framework, let $f$ be an ECA which is applied over a grid of size $n$ and let $\mu$ be a periodic update mode represented as a periodical sequence of subsets of $\integers{n}$ such that
$\mu = (B_0, ..., B_{p-1})$. Function $f_{\mu} = (f, \mu)$, defined $f_{\mu}:\Bn\to\Bn$ defines the \textit{discrete dynamical system} related to ECA $f$ under the update mode $\mu$.
To avoid confusion when it is useful, since ECA $f$ is associated with its Wolfram code, we add this code to the notation as follows: if ECA $f_{\mu}$ is the dynamical system of the rule having Wolfram's code $w$ under $\mu$, we will denote it as $f_{\mu,w}$. {Therefore, for $x\in\Bn$, $f_{\mu,w}(x)=y$ denotes the configuration obtained after applying rule $w$ to all the cells of $x$ according to the order established by $\mu$.}


Let $x\in\Bn$ be a configuration of $f_{\mu}$. The \emph{orbit} of $x$ is the sequence of configurations which results from applying successively $f_{\mu}$ to $x$, {defining} the set $\O(x)$. 
Since $f$ is defined over a grid of finite size and the boundary condition is periodic, the temporal evolution of $x$ governed by the successive applications of $f_{\mu}$ must lead it to eventually enter into a \emph{limit phase}, \ie, a cyclic sequence $\C(x)$ such that for all $y = f_{\mu}^k(x) \in \C(x),$ there  exists $t \in \N, f_{\mu}^t(y) = y$, with $k \in \N$.
The orbit is then separated into two phases: The limit phase and the \emph{transient phase} which corresponds to the finite sequence $(x^0=x,x^1=f_\mu(x),\dots,x^{\ell}=f^{\ell}_\mu(x))$ of length $\ell$ such that for all $i \in \integers{\ell + 1},$ there does not exist $t \in \N, x^{i + t} = x^{i}$. 
The \emph{limit set} of $x$ is the set of configurations belonging to $\C(x)$. In other words, the limit set of $x$ is the asymptotic orbit with respect to $t$ of $x$. If the cardinality of the limit set is equal to 1, we call it a \textit{fixed point}. {An \emph{isolated fixed point} is a fixed point that cannot be reached from any other configuration.} 

In the context of ECAs, it is convenient to represent orbits by \emph{space-time diagrams} which give a visual aspect of the latter, as illustrated in Figure~\ref{fig:space-time}. 

For some of the proofs, we need to use the following specific notations.

Let $x \in \B^n$ be a configuration and $[i,j] \subseteq \integers{n}$ be a subset of cells. 
We denote by $x_{[i,j]}$ the projection of $x$ on $[i,j]$. 
Since we work on ECAs, such a projection defines a sub-configuration and can be 
of three kinds: 
\begin{itemize}
    \item $i < j$ and $x_{[i,j]} = (x_i, x_{i+1}, \dots, x_{j-1}, x_{j})$,
    \item or $i = j$ and $x_{[i,j]} = (x_i)$,
    \item or $i > j$ and $x_{[i,j]} = (x_i, x_{i+1}, \dots, x_n, 
x_0, \dots, x_{j-1}, x_j$).
\end{itemize}

Thus, given $x \in \B^n$ and $i \in \integers{n}$, an ECA with local rule $f$ can be rewritten as 
$$F(x) = (f(x_{[n-1,1]}), f(x_{[0,2]}), \dots, f(x_{[i,i+2]}), \dots, f(x_{[n-3,n-1]}),  
f(x_{[n-2,0]})\text{.}$$

Now, we introduce the notation {$x^{i.j}$}: let $x$ be the initial configuration, and let $i,j\in\N$, $x^{i.j}$ is the configuration obtained starting from $x$ after applying a rule $f$ for $i$ steps and $j$ substeps. For example, following this notation, the initial configuration $x$ is $x=x^0=x^{0.0}$, the configuration after the first substep is $x^{0.1}$, and the configuration after two steps and three substeps is $x^{2.3}$.

\begin{figure}[t]
	\centerline{
	\begin{minipage}{.45\textwidth}
	\centerline{
	\scalebox{1}{\begin{tikzpicture}[help lines/.style={draw=gray!50},
		every node/.style={help lines,rectangle,minimum size=3mm},
		ca/.style={draw=none,row sep=0mm,column sep=0mm, ampersand 
		replacement=\&},
		st0/.style={fill=white,help lines},
		stg/.style={fill=gray,help lines},
		st1/.style={fill=black,help lines}]

		\matrix[ca] {
\node[st1]{};\& \node[st0]{};\& \node[st0]{};\& \node[st1]{};\& \node[st0]{};\& \node[st0]{};\& \node[st0]{};\& \node[st0]{};\\
\node[st1]{};\& \node[st1]{};\& \node[st0]{};\& \node[st1]{};\& \node[st1]{};\& \node[st0]{};\& \node[st0]{};\& \node[st0]{};\\
\node[st1]{};\& \node[st0]{};\& \node[st0]{};\& \node[st1]{};\& \node[st0]{};\& \node[st1]{};\& \node[st0]{};\& \node[st0]{};\\
\node[st1]{};\& \node[st1]{};\& \node[st0]{};\& \node[st1]{};\& \node[st0]{};\& \node[st1]{};\& \node[st1]{};\& \node[st0]{};\\
\node[st1]{};\& \node[st0]{};\& \node[st0]{};\& \node[st1]{};\& \node[st0]{};\& \node[st1]{};\& \node[st0]{};\& \node[st0]{};\\
\node[st1]{};\& \node[st1]{};\& \node[st0]{};\& \node[st1]{};\& \node[st0]{};\& \node[st1]{};\& \node[st1]{};\& \node[st0]{};\\
		};
	\end{tikzpicture}}}
	\end{minipage}
	\quad
	\begin{minipage}{.45\textwidth}
	\centerline{
	\scalebox{1}{\begin{tikzpicture}[help lines/.style={draw=gray!50},
		every node/.style={help lines,rectangle,minimum size=3mm},
		ca/.style={draw=none,row sep=0mm,column sep=0mm, ampersand 
		replacement=\&},
		st0/.style={fill=white,help lines},
		stg/.style={fill=gray,help lines},
		st1/.style={fill=black,help lines}]
        
		\matrix[ca] {
\node[st1]{};\& \node[st0]{};\& \node[st0]{};\& \node[st1]{};\& \node[st0]{};\& \node[st0]{};\& \node[st0]{};\& \node[st0]{};\\
\node[st1]{};\& \node[st1]{};\& \node[st0]{};\& \node[st1]{};\& \node[st1]{};\& \node[st0]{};\& \node[st1]{};\& \node[st0]{};\\
\node[st1]{};\& \node[st0]{};\& \node[st0]{};\& \node[st1]{};\& \node[st0]{};\& \node[st0]{};\& \node[st1]{};\& \node[st0]{};\\
\node[st1]{};\& \node[st1]{};\& \node[st0]{};\& \node[st1]{};\& \node[st1]{};\& \node[st0]{};\& \node[st1]{};\& \node[st0]{};\\
\node[st1]{};\& \node[st0]{};\& \node[st0]{};\& \node[st1]{};\& \node[st0]{};\& \node[st0]{};\& \node[st1]{};\& \node[st0]{};\\
\node[st1]{};\& \node[st1]{};\& \node[st0]{};\& \node[st1]{};\& \node[st1]{};\& \node[st0]{};\& \node[st1]{};\& \node[st0]{};\\
		};
	\end{tikzpicture}}}
	\end{minipage}}
	
	\centerline{
	\begin{minipage}{.45\textwidth}
	\centerline{(a)}
	\end{minipage}
	\quad
	\begin{minipage}{.45\textwidth}
	\centerline{(b)}
	\end{minipage}
	}
	\caption{Space-time diagrams (time going downward) of configuration $10010000$ following Rule $28$ (a) and Rule $29$ (b) under the update mode $\mu = (7,6,5,4,3,2,1,0)$. We are only showing the resulting configurations after full steps have been completed. In this example, $01$ is a wall for both Rule $28$ and Rule $29$. The example starts with two isles of $0$s separated by one $1$. A new isle is created at time step $2$ for case (a) and at time step $1$ for case (b).}
	\label{fig:example_wall_and_isle}
\end{figure}

\begin{definition}[Wall and isle]
Abusing notations, the word $u \in \mathbb{B}^k$, with $k\leq n$, is called a \emph{wall} for a given dynamical system if for all $a, b \in \B$, $f(aub) = u$. In this work we assume that walls are of size $2$, i.e. $k = 2$, unless otherwise stated. Additionally, a configuration $x$ is said to have an \emph{isle} of $1$s (resp. an isle of $0$s) if 
there exists an interval $I = [a,b] \subsetneq \integers{n}$ such that $x_i = 1$ 
(resp. $x_i = 0$) for all $i \in I$ and $x_i = 0$ (resp. $x_i = 1$) for {$i\in\{a-1,b+1\}$. Figure~\ref{fig:example_wall_and_isle} has visual representations of both concepts}.
\end{definition}

\begin{definition}[Universal update mode]\label{def:universal}
If for a given $f$ there exists an update mode $\mu$ that allows all $x\in\Bn$ to reach a fixed point, then $\mu$ is \emph{universal} for $f$.
\end{definition}

\begin{definition}[Covering]\label{def:covering}
A set $S=\{\mu_1,\dots,\mu_k\}$ is a \emph{covering} for $f$ if for any $x\in\Bn$ there exists $\mu\in S$ such that there is a time $t$ in which $f_{\mu}^t(x)$ is a fixed point. 
\end{definition}

\begin{definition}[$X$-Universal update mode]\label{def:almost_universal}
If for a given $f$ there exists an update mode $\mu$ that allows all $x\in X\subsetneq\Bn$ to reach a fixed point, then $\mu$ is \emph{$X$-universal} for $f$.
\end{definition}

\begin{definition}[$X$-Covering]\label{def:partial_covering}
A set $S=\{\mu_1,\dots,\mu_k\}$ is an \emph{$X$-covering} for $f$ if for any $x\in X$ with $X\subsetneq\Bn$ there exists $\mu\in S$ such that there is a time $t$ in which $f_{\mu}^t(x)$ is a fixed point.
\end{definition}

{Note that this means that if there exists a covering (resp. $X$-covering) for $f$ whose cardinality is equal to $1$, then the update mode in such a covering is universal (resp. $X$-universal).}

The following theorem states conditions on an ECA that guarantee its convergence to some fixed point under asynchronous updating, where a neighborhood of an asynchronous cellular automaton is considered \textit{active} if a rule $f$ flips the state of the cell (1 to 0 or 0 to 1). Otherwise, the neighborhood is \textit{passive}. 

\begin{thm}[\cite{das2013,Sethi2016}]\label{thm:DAS}
    Let $f$ be an ECA. Then there exists (at least) one asynchronous update mode $\mu$ for $f$ such that $f_{\mu}$ converges to fixed points if any of the following condition is verified:
    \begin{enumerate}
        \item[(i)] $000$ (resp. $111$) is passive and $010$ (resp. $101$) is active.
        \item[(ii)] $000$, $001$, $010$ and $100$ (resp. $011$, $101$, $110$ and $111$) are passive and $011$ or $110$ (resp. $001$ or $100$) are active.
        \item[(iii)] Either $001$, $010$, $100$ and $101$, or $010$, $011$, $101$ and $110$ are passive.
    \end{enumerate}
\end{thm} 

{Remark that, in what follows, we present results related to sequential update modes, which are a specific kind of asynchronism. As such, our results are stronger than those of~\cite{das2013,Sethi2018size}, in the sense that ours induce the asynchronous case, while the results associated with asynchronism do not necessarily imply the ones associated with sequential update modes.}

\section{Results}
\label{sec:results}
In this section we present our results, starting with the 50 rules that fulfill the conditions of Theorem~\ref{thm:DAS} and later, among the 38 rules that fall outside the purview of said theorem, we will touch upon the ones which admit fixed points. {Note that in the paper, they consider $146$ rules, whilst we will focus on the minimal representative rules, which are $50$.} 

\subsection{Rules for which Theorem~\ref{thm:DAS} holds.}

We will begin this subsection studying the rules that can reach a fixed point from all initial configurations using only one sequential update mode, then we will show the rules that have a covering of sequential update modes. Additionally, since Rules $90$ {and $106$ are} the only ones that fulfill a condition of Theorem~\ref{thm:DAS} that {do not always converge under sequential update modes, they will be given special attention.}

Observe that it was proven in~\cite{Donoso2025} that 16 of the 50 non-equivalent rules that fulfill conditions of the theorem always reach fixed points regardless of the update mode and that 13 of them lead to fixed points under any sequential update mode. Now, we will present the results associated with the remaining 21 rules.

\subsubsection*{Rules that admit a Universal Sequential Update Mode.}
{Remark that the theorems in this section give us a more specific result than the one in~\cite{Sethi2016}. Indeed, in that paper the theorem states that there is \emph{some} way of updating the cells such that a fixed point can be reached (which can vary depending on the initial configuration state) for the rules that fulfill (at least) one of the conditions of the theorem. Meanwhile, here we prove that for 12 ECAs there exists a sequential update mode under which \emph{all} initial configurations can reach a fixed point, while for 7 other elementary rules we prove that any initial configuration state can converge under a sequential update mode.}

\begin{figure}[t]
	\centerline{
	\begin{minipage}{.25\textwidth}
	\centerline{
	\scalebox{1}{\begin{tikzpicture}[help lines/.style={draw=gray!50},
		every node/.style={help lines,rectangle,minimum size=3mm},
		ca/.style={draw=none,row sep=0mm,column sep=0mm, ampersand 
		replacement=\&},
		st0/.style={fill=white,help lines},
		stg/.style={fill=gray,help lines},
		st1/.style={fill=black,help lines}]

		\matrix[ca] {
\node[st0]{};\& \node[st1]{};\& \node[st0]{};\& \node[st1]{};\& \node[st0]{};\& \node[st0]{};\& \node[st1]{};\& \node[st1]{};\\
\node[st0]{};\& \node[stg]{};\& \node[st0]{};\& \node[stg]{};\& \node[st0]{};\& \node[st0]{};\& \node[stg]{};\& \node[st0]{};\\
\node[st0]{};\& \node[stg]{};\& \node[st0]{};\& \node[stg]{};\& \node[st0]{};\& \node[st0]{};\& \node[st0]{};\& \node[st0]{};\\
\node[st0]{};\& \node[stg]{};\& \node[st0]{};\& \node[stg]{};\& \node[st0]{};\& \node[st0]{};\& \node[st0]{};\& \node[st0]{};\\
\node[st0]{};\& \node[stg]{};\& \node[st0]{};\& \node[stg]{};\& \node[st0]{};\& \node[st0]{};\& \node[st0]{};\& \node[st0]{};\\
\node[st0]{};\& \node[stg]{};\& \node[st0]{};\& \node[st0]{};\& \node[st0]{};\& \node[st0]{};\& \node[st0]{};\& \node[st0]{};\\
\node[st0]{};\& \node[stg]{};\& \node[st0]{};\& \node[st0]{};\& \node[st0]{};\& \node[st0]{};\& \node[st0]{};\& \node[st0]{};\\
\node[st0]{};\& \node[st0]{};\& \node[st0]{};\& \node[st0]{};\& \node[st0]{};\& \node[st0]{};\& \node[st0]{};\& \node[st0]{};\\
\node[st0]{};\& \node[st0]{};\& \node[st0]{};\& \node[st0]{};\& \node[st0]{};\& \node[st0]{};\& \node[st0]{};\& \node[st0]{};\\
		};
	\end{tikzpicture}}}
	\end{minipage} 
	\begin{minipage}{.24\textwidth}
	\centerline{
	\scalebox{1}{\begin{tikzpicture}[help lines/.style={draw=gray!50},
		every node/.style={help lines,rectangle,minimum size=3mm},
		ca/.style={draw=none,row sep=0mm,column sep=0mm, ampersand 
		replacement=\&},
		st0/.style={fill=white,help lines},
		stg/.style={fill=gray,help lines},
		st1/.style={fill=black,help lines}]
        
		\matrix[ca] {
\node[st0]{};\& \node[st1]{};\& \node[st0]{};\& \node[st1]{};\& \node[st0]{};\& \node[st0]{};\& \node[st1]{};\& \node[st1]{};\\
\node[st0]{};\& \node[st0]{};\& \node[st0]{};\& \node[stg]{};\& \node[st0]{};\& \node[st0]{};\& \node[stg]{};\& \node[stg]{};\\
\node[st0]{};\& \node[st0]{};\& \node[st0]{};\& \node[st0]{};\& \node[st0]{};\& \node[st0]{};\& \node[stg]{};\& \node[stg]{};\\
\node[st0]{};\& \node[st0]{};\& \node[st0]{};\& \node[st0]{};\& \node[st0]{};\& \node[st0]{};\& \node[stg]{};\& \node[st0]{};\\
\node[st0]{};\& \node[st0]{};\& \node[st0]{};\& \node[st0]{};\& \node[st0]{};\& \node[st0]{};\& \node[st0]{};\& \node[st0]{};\\
\node[st0]{};\& \node[st0]{};\& \node[st0]{};\& \node[st0]{};\& \node[st0]{};\& \node[st0]{};\& \node[st0]{};\& \node[st0]{};\\
\node[st0]{};\& \node[st0]{};\& \node[st0]{};\& \node[st0]{};\& \node[st0]{};\& \node[st0]{};\& \node[st0]{};\& \node[st0]{};\\
\node[st0]{};\& \node[st0]{};\& \node[st0]{};\& \node[st0]{};\& \node[st0]{};\& \node[st0]{};\& \node[st0]{};\& \node[st0]{};\\
\node[st0]{};\& \node[st0]{};\& \node[st0]{};\& \node[st0]{};\& \node[st0]{};\& \node[st0]{};\& \node[st0]{};\& \node[st0]{};\\
		};
	\end{tikzpicture}}}
	\end{minipage}
	\begin{minipage}{.24\textwidth}
	\centerline{
	\scalebox{1}{\begin{tikzpicture}[help lines/.style={draw=gray!50},
		every node/.style={help lines,rectangle,minimum size=3mm},
		ca/.style={draw=none,row sep=0mm,column sep=0mm, ampersand 
		replacement=\&},
		st0/.style={fill=white,help lines},
		stg/.style={fill=gray,help lines},
		st1/.style={fill=black,help lines}]

		\matrix[ca] {
\node[st1]{};\& \node[st1]{};\& \node[st0]{};\& \node[st1]{};\& \node[st1]{};\& \node[st1]{};\& \node[st0]{};\& \node[st0]{};\\
\node[stg]{};\& \node[stg]{};\& \node[st0]{};\& \node[stg]{};\& \node[stg]{};\& \node[stg]{};\& \node[st0]{};\& \node[stg]{};\\
\node[stg]{};\& \node[stg]{};\& \node[st0]{};\& \node[stg]{};\& \node[stg]{};\& \node[stg]{};\& \node[st0]{};\& \node[stg]{};\\
\node[stg]{};\& \node[stg]{};\& \node[st0]{};\& \node[stg]{};\& \node[stg]{};\& \node[st0]{};\& \node[st0]{};\& \node[stg]{};\\
\node[stg]{};\& \node[stg]{};\& \node[st0]{};\& \node[stg]{};\& \node[st0]{};\& \node[st0]{};\& \node[st0]{};\& \node[stg]{};\\
\node[stg]{};\& \node[stg]{};\& \node[st0]{};\& \node[st0]{};\& \node[st0]{};\& \node[st0]{};\& \node[st0]{};\& \node[stg]{};\\
\node[stg]{};\& \node[stg]{};\& \node[st0]{};\& \node[st0]{};\& \node[st0]{};\& \node[st0]{};\& \node[st0]{};\& \node[stg]{};\\
\node[stg]{};\& \node[st0]{};\& \node[st0]{};\& \node[st0]{};\& \node[st0]{};\& \node[st0]{};\& \node[st0]{};\& \node[stg]{};\\
\node[st0]{};\& \node[st0]{};\& \node[st0]{};\& \node[st0]{};\& \node[st0]{};\& \node[st0]{};\& \node[st0]{};\& \node[st1]{};\\
		};
	\end{tikzpicture}}}
	\end{minipage}
	\begin{minipage}{.25\textwidth}
	\centerline{
	\scalebox{1}{\begin{tikzpicture}[help lines/.style={draw=gray!50},
		every node/.style={help lines,rectangle,minimum size=3mm},
		ca/.style={draw=none,row sep=0mm,column sep=0mm, ampersand 
		replacement=\&},
		st0/.style={fill=white,help lines},
		stg/.style={fill=gray,help lines},
		st1/.style={fill=black,help lines}]
		\matrix[ca] {
\node[st1]{};\& \node[st1]{};\& \node[st0]{};\& \node[st1]{};\& \node[st1]{};\& \node[st1]{};\& \node[st0]{};\& \node[st0]{};\\
\node[stg]{};\& \node[stg]{};\& \node[st0]{};\& \node[stg]{};\& \node[stg]{};\& \node[st0]{};\& \node[st0]{};\& \node[st0]{};\\
\node[stg]{};\& \node[stg]{};\& \node[st0]{};\& \node[stg]{};\& \node[st0]{};\& \node[st0]{};\& \node[st0]{};\& \node[st0]{};\\
\node[stg]{};\& \node[stg]{};\& \node[st0]{};\& \node[st0]{};\& \node[st0]{};\& \node[st0]{};\& \node[st0]{};\& \node[st0]{};\\
\node[stg]{};\& \node[st0]{};\& \node[st0]{};\& \node[st0]{};\& \node[st0]{};\& \node[st0]{};\& \node[st0]{};\& \node[st0]{};\\
\node[st0]{};\& \node[st0]{};\& \node[st0]{};\& \node[st0]{};\& \node[st0]{};\& \node[st0]{};\& \node[st0]{};\& \node[st0]{};\\
\node[st0]{};\& \node[st0]{};\& \node[st0]{};\& \node[st0]{};\& \node[st0]{};\& \node[st0]{};\& \node[st0]{};\& \node[st0]{};\\
\node[st0]{};\& \node[st0]{};\& \node[st0]{};\& \node[st0]{};\& \node[st0]{};\& \node[st0]{};\& \node[st0]{};\& \node[st0]{};\\
\node[st0]{};\& \node[st0]{};\& \node[st0]{};\& \node[st0]{};\& \node[st0]{};\& \node[st0]{};\& \node[st0]{};\& \node[st0]{};\\
		};
	\end{tikzpicture}}}
	\end{minipage}}
	
	\centerline{
	\begin{minipage}{.25\textwidth}
	\centerline{(a)}
	\end{minipage}
	\begin{minipage}{.24\textwidth}
	\centerline{(b)}
	\end{minipage}
	\begin{minipage}{.24\textwidth}
	\centerline{(c)}
	\end{minipage}
	\begin{minipage}{.25\textwidth}
	\centerline{(d)}
	\end{minipage}
	}
	\caption{Space-time diagrams (time going downward) of configuration 
		$01010011$, in (a) and (b), and $11011100$, in (c) and (d), following Rule $138$ under: 
		(a) and (c) the update mode $\mu_{1} = (9,8,7,6,5,4,3,2,1,0)$, 
		(b) the update mode $\mu_{2} = (1,3,9,8,0,2,4,5,6,7)$, 
		(d) the update mode $\mu_{3} = (5,4,3,1,0,2,6,7,8,9)$. Configurations colored in gray (resp. black) represent the ones obtained at substeps (resp. steps). Update mode $\mu_1$ is the one defined in Theorem~\ref{thm:universal}, while update modes $\mu_2$ and $\mu_3$ are obtained following the instructions given by the proof of Theorem~\ref{thm:DAS} in~\cite{Sethi2016}.}
	\label{fig:example}
\end{figure}

{In Figure~\ref{fig:example} we illustrate the difference between our results and the ones obtained in~\cite{Sethi2016}, presenting four example dynamics with local rule $138$, starting from two configurations, updated following $\mu_1$ (the update mode defined in Theorem~\ref{thm:universal}) and $\mu_2$ and $\mu_3$ (the update modes obtained according to the proof presented in~\cite{Sethi2016}. Notice that the update modes $\mu_2$ and $\mu_3$ have the advantage of reaching the fixed point in a smaller number of iterations (in particular, case (c) does not reach a fixed point after one step, needing two complete substeps to converge), while the update mode $\mu_1$ has the advantage of leading \emph{all} initial configurations to a fixed point for the rules of Theorem~\ref{thm:universal}.}

\begin{table}[t]
\begin{center}
\begin{tabular}{|c|c|c|c|c|c|c|c|c|}
    \hline
         &111&110&101&100&011&010&001&000  \\\hline
         2  &0&0&0&0&0&0&1&0\\\hline 
         10 &0&0&0&0&1&0&1&0\\\hline 
         26 &0&0&0&1&1&0&1&0\\\hline 
         34 &0&0&1&0&0&0&1&0\\\hline 
         42 &0&0&1&0&1&0&1&0\\\hline 
         58 &0&0&1&1&1&0&1&0\\\hline 
         130&1&0&0&0&0&0&1&0\\\hline 
         138&1&0&0&0&1&0&1&0\\\hline 
         154&1&0&0&1&1&0&1&0\\\hline 
         162&1&0&1&0&0&0&1&0\\\hline 
         170&1&0&1&0&1&0&1&0\\\hline 
    \end{tabular}    
    \caption{Definition of rules $2$, $10$, $26$, $34$, $42$, $58$, $130$, $138$, $154$, $162$ and $170$.}
    \label{tab:rules_universal}
\end{center}
\end{table}

\begin{thm}\label{thm:universal}
    The sequential update mode $\mu=(n-1,n-2,\dots,0)$ is universal for an ECA rule R if any of the following is true:
    \begin{itemize}
        \item $000$ and $010$ are passive, and $001$, $010$ and $110$ are active. 
        \item $000$ and $011$ are passive, and $001$, $010$, $100$ and $110$ are active.
    \end{itemize}
\end{thm}

\begin{proof}
Let $\mu=(n-1,n-2,\dots,0)$.
Let us study each of the rules that fulfill either condition.
    \begin{itemize}
        \item Rule 2: By definition (see Table~\ref{tab:rules_universal}), only $f_2(001)=1$. Since we are updating the cells of the configuration one-by-one from right to left (by definition of $\mu$), the state of all cells immediately becomes $0$, unless the state of the first cell of the configuration is $1$ and the last two have state equal to $0$: $x^{0}(=x)=1(1^{b_1}0^{a_1})1^{b_2}\dots0^{a_{k-1}}(1^{b_k}0^{a_k})00$, such that $\sum_{i=1}^{k}a_i+b_i=n-3$. In this case, after the first substep, configuration $x^{0.0}$ becomes $x^{0.1}=1(1^{b_1}0^{a_1})1^{b_2}\dots0^{a_{k-1}}(1^{b_k}0^{a_k})(01)$. Then, after $a_k$ substeps, $x^{0.1}$ becomes $x^{0.a_k+1}=1(1^{b_1}0^{a_1})1^{b_2}\dots0^{a_{k-1}}(1^{b_k}01^{a_k+2})$. Making another substep does not produce any change, since $f_2(101)=0$, so $x^{0.a_k+2}=x^{0.a_k+1}$.
        
        Then, by definition, the remaining substeps of the first step will make all the cells from $0$ to $n-a_{k}-2$ become $0$ so that $x^1=0^{n-a_k-2}1^{a_k+2}$. Finally, the first $a_k+2$ substeps of the second step lead the configuration to the fixed point $0^n$, since $f_2(110)=0$.
        \item Rule 10: The analysis is the same as that for Rule $2$, but now it includes a second type of configuration $x_1=10^{a_1}0^{b_1}\dots0^{a_k}1^{b_k}01$ whose orbits have a  transient part longer than zero.
        \item Rule 34: The analysis is the same as that for Rule $2$, adding $x_2=10^{a_1}0^{b_1}\dots0^{a_k}1^{b_k}10$ as another type of configuration whose orbits have a  transient part longer than zero.
        \item Rule 42: The analysis is the same as that for Rule $2$, also including $x_1$ from Rule 10 and $x_2$ from Rule 34 as configurations whose orbits have a  transient part longer than zero.
        \item Rule 130: It is identical to Rule 2, except that now the configuration $1^n$ is also a fixed point, but with an empty basin of attraction.
        \item Rule 138: It is identical to Rule 10, except that now the configuration $1^n$ is also a fixed point, but with an empty basin of attraction.
        \item Rule 162: It is identical to Rule 34, except that now configuration $1^n$ is also a fixed point and configurations $y_1=10^{n-1}$, $y_2=1^20^{n-2}$, \dots, $y_{n-1}=1^{n-1}0$ lead to it.
        \item Rule 170: It is identical to Rule 42, except that now configuration $1^n$ is also a fixed point. It is easy to see that the configurations lead to $1^n$ or $0^n$ depending on if the state of the first cell is $1$ or $0$ respectively.
        \item Rule 26: We need to do a case-by-case analysis:
        \begin{equation*}
            \begin{split}
                x = 1^{n-1}0 &\rightarrow f_{\mu,26}(x)=0^n\text{,}\\
                x = 1^{n-2}0^{2}\rightarrow f_{\mu,26}(x)=0^{n-1}1 &\rightarrow f_{\mu,26}^2(x)=0^n\text{,}\\
                x = 1^{n-3}0^{3}\rightarrow f_{\mu,26}(x)=0^{n-2}1^2 &\rightarrow f_{\mu,26}^2(x)=0^n\text{,}\\
                \vdots\\
                x = 1^{n-k}0^{k}\rightarrow f_{\mu,26}(x)=0^{n-k+1}1^{k-1} &\rightarrow f_{\mu,26}^2(x)=0^n\text{,}\\
            \end{split}
        \end{equation*}
        Now, if we have an isle of $1$s surrounded by $0$s: $x=0^{k_1}1^{k_2}0^{n-k_1-k_2}$
        Then, we will have:      
        \begin{equation*}
            \begin{split}
                f_{\mu,26}(x) &= 0^{k_1+k_2}10^{n-(k_1+k_2)-1}\text{,}\\
                f_{\mu,26}^2(x) &= 1^{k_1+k_2+2}0^{n-(k_1+k_2+2)}\text{,}\\
                f_{\mu,26}^4(x)&=0^n\text{.}\\
            \end{split}
        \end{equation*}
        The analysis is analogous for $x=1^{k_1}0^{k_2}1^{n-k_1-k_2}$. Now, let us consider $x=0^{a_1}1^{b_1}\dots0^{a_k}1^{b_k}0^{a_{k+1}}$, then 
        $$f_{\mu,26}(x) = 0^{a_1+b_1}10^{a_2+b_2-1}1\dots 0^{a_k+b_k-1}10^{a_{k+1}-1}\text{.}$$
        If $k$ is odd, then
        $$f_{\mu,26}^2(x) = 1^{a_1+b_1+2}0^{a_2+b_2}\dots 1^{a_k+b_k}10^{a_{k+1}-2}\text{.}$$
        If $k$ is even, then
        $$f_{\mu,26}^2(x) = 0^{a_1+b_1+2}1^{a_2+b_2}0^{a_3+b_3}\dots 1^{a_k+b_k}10^{a_{k+1}-2}\text{.}$$
        In both cases, the number of isles decreases by half every two iterations until we are left with only one isle of $1$s and one of $0$s, which we already know leads to a fixed point.
        \item Rule 58: The analysis follows a similar reasoning as that for Rule 26.
        \item Rule 154: The analysis follows a similar reasoning as that for Rule 26, but in this case $1^n$ is a fixed point with an empty basin of attraction.
    \end{itemize}
\end{proof}

\begin{thm}
    For Rule $24$ the sequential update mode $\mu=(0,1,\dots,n-1)$ is universal.
\end{thm}
\begin{proof}
    The reasoning follows the lines of that for Rule $10$, up to the update mode inversion.
\end{proof}
\subsubsection{Rules that admit a Covering of Sequential Update Modes}
\begin{table}[t]
\begin{center}
\begin{tabular}{|c|c|c|c|c|c|c|c|c|}
    \hline
         &111&110&101&100&011&010&001&000  \\\hline
         104&0&1&1&0&1&0&0&0\\\hline 
    \end{tabular}    
    \caption{Definition of Rule $104$.}
    \label{tab:104}
\end{center}
\end{table}
\begin{thm}\label{thm:104_almost_universal}
    For Rule $104$, if the size of the ring is even, then the sequential update mode $\mu=(n-1,n-2,\dots,0)$ is $E$-universal.
\end{thm}
\begin{proof}
Note that, by definition of Rule $104$, the word $w=00$ is a wall (see Table~\ref{tab:104}). Let $\mu=(n-1,n-2,\dots,0)$.

If $n$ is odd, the update mode $\mu$ cannot lead all initial configurations to a fixed point, that is, there exists at least one initial configuration that leads to a cycle of length $T\geq2$.

Indeed, let us consider $x=(01)1^{n-2}$, with $n$ an odd number.
\begin{equation*}
\begin{split}
f_{\mu,104}(x)=&111(01)^{\frac{n-3}{2}}\text{,}\\
f_{\mu,104}^2(x)=&(01)^21^{n-4}\text{,}\\
f_{\mu,104}^3(x)=&11111(01)^{\frac{n-5}{2}}\text{,}\\
f_{\mu,104}^4(x)=&(01)^{3}1^{n-6}\text{,}\\
\vdots&\\
f_{\mu,104}^t(x)=&1^{t+2}(01)^{\frac{n-5}{2}}\text{, with $t$ an odd number,}\\
f_{\mu,104}^{t+1}(x)=&(01)^{\frac{t+3}{2}}1^{n-(t+2)}\text{.}\\
\end{split}
\end{equation*}    
After $t^{\star}=n-4$ steps, we have
\begin{equation*}
\begin{split}
f_{\mu,104}^{t^{\star}}(x)=&1^{n-2}(01)\text{,}\\
f_{\mu,104}^{t^{\star}+1}(x)=&(01)^{\frac{n-1}{2}}1\text{,}\\
f_{\mu,104}^{t^{\star}+2}(x)=&1^n\text{,}\\
f_{\mu,104}^{t^{\star}+3}(x)=&1(01)^{\frac{n-1}{2}}\text{,}\\
f_{\mu,104}^{t^{\star}+4}(x)=&(10)1^{n-2}\text{.}\\
\end{split}
\end{equation*}
Let $y=f_{\mu,104}^{t^{\star}+4}(x)=x^{n}$. We have
    \begin{equation*}
\begin{split}
f_{\mu,104}(y)=&111(10)^{\frac{n-3}{2}}\text{,}\\
f_{\mu,104}^2(y)=&(10)^21^{n-4}\text{,}\\
\vdots&\\
f_{\mu,104}^t(y)=&1^{t+2}(10)^{\frac{n-5}{2}}\text{, with $t$ an odd number,}\\
f_{\mu,104}^{t+1}(y)=&(10)^{\frac{t+3}{2}}1^{n-(t+2)}\text{,}\\
\vdots&\\
f_{\mu,104}^{t^{\star}}(y)=&1^{n-2}(10)\text{, with }t^{\star}=n-4\\
f_{\mu,104}^{t^{\star}+1}(y)=&(10)^{\frac{n-1}{2}}1\text{,}\\
f_{\mu,104}^{t^{\star}+2}(y)=&01^{n-1}=x\text{.}\\
\end{split}
\end{equation*}  
Thus we have found a configuration $x$ that does not reach a fixed point under $f_{\mu,104}$.

    Now, let us consider an initial configuration $x=0^{a_1}1^{b_1}\dots 0^{a_k}1^{b_k}$, such that $\sum_{i=1}^k a_i+b_i=n\equiv 0\mod2$.

    If $a_i=1$, with  $i\in\{1,\dots,k\}$, those isolated $0$s will disappear (because $f_{104}(101)=1$). If $a_i \geq 2$, with  $i\in\{1,\dots,k\}$, those isles of $0$ cannot decrease in size because they contain walls.

    If $b_i=1$, with  $i\in\{1,\dots,k\}$, those isolated $1$s will disappear (because $f_{104}(010)=0$). 
    Consider the case in which $b_i=2$, with  $i\in\{1,\dots,k\}$. Since $f_{104}(011)=f_{104}(110)=1$, we could have a wall depending on the number of $0$s at either side of the isle of $1$s. We know that the $0$s to the right of the isle of $1$s will not decrease. 
    If there is only one $0$ to the left of the isle of $1$s of size two, then we already know that the isolated $0$s disappear and after this iteration there will be an isle of $1$s of size at least 4.

    Finally, in the case where $b_i\geq 3$, with  $i\in\{1,\dots,k\}$, then the isle of $1$s will be broken up because $f_{104}(111)=0$:
    \begin{itemize}
        \item If $b_i$ is even, then after the first iteration this isle of $1$s will become $11(01)^{\frac{b_i-2}{2}}$ and after the second it will be $110^{b_i-2}$. And from there it follows the analysis for $b_i=2$.
        \item If $b_i$ is odd, after the first iteration the isle will become $0(01)^{\frac{b_i-1}{2}}$; and after the second the isolated $1$s will disappear.
    \end{itemize}
Therefore, we have proven that all initial configurations of a ring of even size $n$ arrive at a fixed point under $f_{\mu,104}$.
\end{proof}

In Theorem~\ref{thm:104_almost_universal} we chose an update mode with a simple definition, that is: $\mu=(n-1,n-2,\dots,0)$, to prove that there exists at least one sequential update mode which is universal when the size of the ring is even. Nevertheless, when we run computer simulations for Rule $104$ with a ring of size $n=8$ we found  544 different sequential update modes that lead all configurations to a fixed point. Here we present a few examples: $\mu_1=(0,1,2,3,4,5,6,7)$, $\mu_2=(0,1,3,6,4,5,2,7)$, $\mu_3= (0,2,6,1,3,4,5,7)$ and $\mu_4=(2,6,3,0,1,4,5,7)$.

So far, the rules have had a universal sequential update mode (restricted to rings of even size, in the cases of Rules $104$ and $106$). In what remains of this section, we have six rules that have a covering of sequential update modes and one rule that does not have it.
\begin{thm}
For Rule $104$, if the size of the ring is odd, there exists a covering of sequential update modes. 
\end{thm}
\begin{proof}
Let $x=0^{a_1}1^{b_1}\dots0^{a_k}1^{b_k}$ with $\sum_{i=1}^k a_i+b_i=n$.

By definition, we know that the configuration $0^{n}$ is a fixed point, and that any configuration where $b_i=2$ for all $i\in\{1,\dots,k\}$, is also a fixed point, since $f_{104}(110)=f_{104}(011)=1$ (see Table~\ref{tab:104}).

To reach a fixed point from any initial configuration such as $x$, we can start by dividing the isles of $1$s where $b_i\geq 3$, using that $f_{104}(111)=0$, leaving only $1$s that are isolated or in pairs,~\ie, $x'=0^{a_1'}1^{b_1'}\dots0^{a_{k'}'}1^{b_{k'}'}$ with $\sum_{i=1}^{k'} a_i'+b_i'=n$, where $b_i'\in\{1,2\}$, for all $i\in\{1,\dots,k'\}$.

Then, we can update all isolated $1$s, using that $f_{104}(010)=0$. This results in the configuration $x''=0^{c_1}1^{d_1}\dots0^{c_{\ell}}1^{d_{\ell}}$ with $\sum_{i=1}^{\ell} c_i+d_i=n$, where $d_i=2$, for all $i\in\{1,\dots,\ell\}$.

Since we have already reached a fixed point, we can now update any cells that have not yet been updating, thus completing a sequential update mode.
\end{proof}

\begin{table}[t]
\begin{center}
\begin{tabular}{|c|c|c|c|c|c|c|c|c|}
    \hline
         &111&110&101&100&011&010&001&000  \\\hline
         18 &0&0&0&1&0&0&1&0\\\hline 
         50 &0&0&1&1&0&0&1&0\\\hline 
         74 &0&1&0&0&1&0&1&0\\\hline 
         122&0&1&1&1&1&0&1&0\\\hline 
         146&1&0&0&1&0&0&1&0\\\hline 
         178&1&0&1&1&0&0&1&0\\\hline 
    \end{tabular}    
    \caption{Definition of rules $18$, $50$, $74$, $122$, $146$ and $178$.}
    \label{tab:rules_coverings}
\end{center}
\end{table}
\begin{thm}\label{thm:coverings}
For Rules $18$, $50$, $74$, $122$, $146$ and $178$ there exists a covering of sequential update modes.
\end{thm}
\begin{proof}
By definition (see Table~\ref{tab:rules_coverings}), we know that $0^n$ is a fixed point for all these rules and additionally $1^n$ is a fixed point for rules $146$ and $178$.

Let $x=0^{a_1}1^{b_1}\dots0^{a_k}1^{b_k}$ with $a_1,b_1,\dots,a_k,b_k\geq0$ such that $\sum_{i=1}^k a_i+b_i+=n$.
\begin{enumerate}
\item Rule $18$: Since $f_{18}(111)=0$ we can start by updating every other cell of the isles of $1$s. The resulting configuration will only have $1$s in pairs or isolated. Then we eliminate the remaining $1$s using that $f_{18}(010)=f_{18}(011)=f_{18}(110)=0$. 
\item Rule $50$: It is identical to Rule $18$.
\item Rule $74$: It is identical to Rule $18$, except that $f_{74}(110)=1$ and thus pairs of $1$s must be updated from right-to-left.
\item Rule $122$: The analysis is the same as for Rule $18$, except that $f_{122}(011)=f_{122}(110)=1$ which means that configurations with pairs of $1$s, such as $x'=0^{c_1}11\dots0^{c_{k-1}}110^{c_k}$ are also fixed points.
\item Rule $146$: Since $f_{146}(011)=f_{146}(110)=0$, we can start by updating the cells of the isles of $1$s from left-to-right (or right-to-left) and given that $f_{146}(010)=0$ the remaining $1$s can be eliminated, thus arriving at $0^n$.
\item Rule $178$: It is identical to Rule $146$.
\end{enumerate}
Note that after eliminating all unwanted $1$s for each of these rules we must still establish an order for updating the cells that belonged to an isle of $0$s in the original configuration, in order to define a sequential update mode.
\end{proof}
\FloatBarrier
\subsubsection{Rules with no Covering}

\begin{table}[t]
\begin{center}
\begin{tabular}{|c|c|c|c|c|c|c|c|c|}
    \hline
         &111&110&101&100&011&010&001&000  \\\hline
         90 &0&1&0&1&1&0&1&0\\\hline 
         106&0&1&1&0&1&0&1&0\\\hline 
    \end{tabular}    
    \caption{Definition of Rules $90$ and $106$.}
    \label{tab:90}
\end{center}
\end{table}
Unlike the rest of the rules that fulfill one of the conditions of Theorem~\ref{thm:DAS}, Rules $90$ and $106$ do not allow all configurations to reach a fixed point, no matter which of the $n!$ sequential update mode is being used.
Indeed, for Rule $90$ (resp. Rule $106$) for rings of size $n=9$ or more, if $n$ is not a power of $2$ (resp. $n$ is not even), there are configurations that cannot reach a fixed point, regardless of the sequential update mode being used. Instead, it appears to be necessary to allow for one cell to be updated more than once per step. {Remark that this does not contradict Theorem~\ref{thm:DAS}~\cite{Sethi2016}, since in it asynchronous update modes are considered, which is a more general class than sequential update modes.}

\begin{thm}\label{thm:rule90}
    Let $\mu$ be any sequential update mode, and let $f_{\mu,90}$ be the discrete dynamical system in which Rule $90$ evolves under the update mode $\mu$. Then, configurations that contain only two $1$s surrounded by $0$s cannot reach a fixed point under sequential update modes.
\end{thm}

\begin{proof}
Let $\mu$ be a sequential update mode.

    Let $x=0^{k_1}110^{k_2}$ such that $k_1+k_2+2=n$. Let $i,j$ be the cells whose state is $1$.

    If the cell to the left of the pair of $1$s updates first, which we will denote $\mu_{i-1}<\mu_{i}$, then we will have three consecutive $1$s. Then we can update $i$ which results in two isolated $1$s and then we update $j$ which leaves us with only one $1$ located in $i-1$. 
    
    Since we want the last one to disappear, we would need that $\mu_{i-2}<\mu_{i-1}$, which would mean that, in the next iteration, cell $i-2$ will change its state to $1$ and we, once again, have two consecutive cells whose state is $1$.

    Analogously, if instead we have that $\mu_{j}<\mu_{j+1}$ we obtain the two consecutive ones located in the cells $j+1$ and $j+1$. Thus, if we follow this strategy, we would need the cell $i-1$ (or $j+1$) to be updated twice before $i-2$ (or $j+2$) can change its state to $1$.

    Similarly, if we force the isle of $1$s to increase to a larger size, we would need that the number of cells whose state will become $1$ to be updated consecutively from left-to-right or right-to-left, depending on the side from which we want to extend the isle, since the only two neighborhoods that increase the number of $1$s are $100$ and $001$ (see Table~\ref{tab:90}). 
    
	But given that we need every other cell to be updated consecutively within an isle for it to disappear (since $f_{90}(111)=0$ and then $f_{90}(010)=0$), using a similar reasoning as before, we can see that the larger isle cannot disappear.
	
Additionally, Rule $90$ under sequential update modes cannot lead a configuration such as $x=0^{k_1}110^{k_2}$ to the fixed points $(011)^{\frac{n}{3}}$, $(101)^{\frac{n}{3}}$ and $(110)^{\frac{n}{3}}$ which, naturally, only exist when the size of the ring is a multiple of $3$. Using a similar reasoning as the previous case, we can see that we would need one order to allow the size of the isle of $1$s to increase and a different order to allow every third cell to be updated, thus creating pairs of $1$s separated by isolated $0$s.

	Therefore, there does not exist a sequential update mode which would allow Rule $90$ to make a configuration composed of an isle of two $1$s surrounded by $0$s delete all $1$s and thus reach the fixed point $0^n$ nor the fixed points $(011)^{\frac{n}{3}}$, $(101)^{\frac{n}{3}}$ and $(110)^{\frac{n}{3}}$.	
	\end{proof}

\begin{corollary}
    There is no covering of sequential update modes for Rule $90$.
\end{corollary}
\begin{proof}
    It is direct from Theorem~\ref{thm:rule90}.
\end{proof}

\begin{thm}
For Rule $106$, if the size of the ring is even, then the sequential update mode $\mu=(n-1,n-2,\dots,0)$ is universal.
\end{thm}
\begin{proof}
Rule 106 allows isles of $0$s to decrease in size (only at the end of the configuration and only if the state of the first cell of the configuration is $1$), in which case a new isle of $1$s appears of the same size as the original isle of $0$s. From there, we follow a similar analysis to the one used in Theorem~\ref{thm:104_almost_universal}. Note that for rings of odd size, we can use the same counterexample for $f_{\mu,106}$ as we did for $f_{\mu,104}$ in said Theorem.
\end{proof}
\begin{thm}
There is no covering of sequential update modes for Rule $106$, when the size of the ring is odd.
\end{thm}
\begin{proof}
It follows a similar reasoning as the one used for Theorem~\ref{thm:rule90}, where configurations that contain only two isolated $1$s cannot reach $0^n$, which is the only fixed point for this rule, when the size of the ring is odd.
\end{proof}

{So far, we have proven that sequential update modes are enough to ensure the convergence of 48 out of the 50 elementary rules that can be obtained from Theorem~\ref{thm:DAS}. Let us now see what happens with the ones outside of it.}

\subsection{Notable Rules Outside the Theorem}
\begin{table}[t]
\resizebox{\textwidth}{!}{
\begin{tabular}{|c|c|c|c|c|c|c|}
\hline
\multicolumn{2}{|c|}{Outside Theorem}&I&II&III&IV&Total\\\hline
\multirow{6}{*}{\makecell{With\\Fixed\\Points}}&\makecell{$E$-Universal}&$\emptyset$&\makecell{7,15}&$\emptyset$&$\emptyset$&2\\\cline{2-7}
&\makecell{$T$-Universal\\(conjecture)}&$\emptyset$&$\emptyset$&\makecell{45}&$\emptyset$&1\\\cline{2-7}
&\makecell{With $E$-Covering}&$\emptyset$&\makecell{23}&30&$\emptyset$&2\\\cline{2-7}
&\makecell{Without\\Covering}&$\emptyset$&\makecell{6,14,28,29,(37),38,\\46,62,73,108,\\134,142,156}&\makecell{22,60,105,\\126,150}&54,110&20\\\hline
\multicolumn{2}{|c|}{\makecell{Without\\Fixed Points}}&$\emptyset$&\makecell{1,3,9,11,19,\\25,27,33,35,\\41,43,51,57}&$\emptyset$&$\emptyset$&13\\\hline
\end{tabular}}
\caption{Summary of the classification of the rules that do not fulfill any  of the conditions stated by Theorem~\ref{thm:DAS}}
\label{tab:outside-theorem}
\end{table}
Theorem~\ref{thm:DAS} provides a thorough characterization of ECAs for which there is always, for each initial configuration, at least one asynchronous update mode that allows it to converge to a fixed point. However, we have determined additional conditions (that depend on the multiplicity of the size of the ring) under which certain rules that do not fulfill any of the conditions of said theorem are still able to reach fixed points.

\FloatBarrier
\subsubsection{Restriction to Rings of Even Size}
\label{sssec:even}
The rules studied in this section only have (non-isolated) fixed points if the size of the ring is even.  Hence, we cannot generalize the results for these rules, but we can restrict our study to the constraints that allow fixed points to exist. Hence, given $E$ the set of configurations of even size (over a grid of even size), the results of this section are related to $E$-universality and $E$-coverings. 

\begin{table}[t]
\begin{center}
\begin{tabular}{|c|c|c|c|c|c|c|c|c|}
    \hline
         &111&110&101&100&011&010&001&000  \\\hline
         7&0&0&0&0&0&1&1&1\\\hline 
         15&0&0&0&0&1&1&1&1\\\hline 
         23&0&0&0&1&0&1&1&1\\\hline 
         30&0&0&0&1&1&1&1&0\\\hline 
    \end{tabular}    
    \caption{Definition of rules $7$, $15$, $23$ and $30$.}
    \label{tab:rules7_15_23}
\end{center}
\end{table}

{The convergence of Rules $7$, $15$, $23$ and $30$ (along with that of Rules $6$, $14$ and $22$) has been studied under asynchronous updating in~\cite{Sethi2018size}, where it was shown that for a given initial configuration of even size, there is an asynchronous way of updating the cells (which depends on the initial conditions) such that a fixed point can be reached. Here, we will show that sequential update modes are enough to ensure convergence and, moreover, that for Rules $7$ and $15$ there is a particular sequential update mode that guarantees that any initial configuration of even size can always reach a fixed point.}

\begin{lemma}\label{lem:7}
For Rules $7$, $15$ and $23$, there are two fixed points for a ring of size $n$, if and only if $n$ is even: $(01)^{\frac{n}{2}}$ and $(10)^{\frac{n}{2}}$. Rule $30$ has the configuration $0^n$ as a fixed point, in addition to the previous two.
\end{lemma}
\begin{proof}
    It is direct from the definitions if these four Rules (see Table~\ref{tab:rules7_15_23}).
\end{proof}
\begin{thm}
    Given a ring of even size $n$, $\mu=(0,1,2,\dots,n-2,n-1)$ is an $E$-universal sequential update mode for Rules $7$ and $15$.
\end{thm}
\begin{proof}
From Lemma~\ref{lem:7}, we know that $x=(01)^{\frac{n}{2}}$ and $y=(10)^{\frac{n}{2}}$ are fixed points.

    Let $x=0^n$ and $\mu=(0,1,2,\dots,n-2,n-1)$. Since $f_{7}(000)=1$, it is easy to calculate that $f_{\mu,7}(x) = (10)^{\frac{n}{2}}$. Similarly, for $x=1^n$ we have that $f^2_{\mu,7}(x)=(01)^{\frac{n}{2}}$.

    Next, let $x=0^{a}1^{b}$ such that $a+b=n$, which means that $a$ and $b$ must either both be even or both be odd.
    
    If they are odd, then:
    \begin{equation*}
        \begin{split}
            f_{\mu,7}(x) =&\hspace{5pt} (01)^{\frac{a-1}{2}}0^b1\text{, and}\\
            f_{\mu,7}^2(x) =&\hspace{5pt} (01)^{\frac{a-1}{2}}(01)^{\frac{b-1}{2}}01\text{.}\\
        \end{split}
    \end{equation*}
    
    If they are even, then:
    \begin{equation*}
        \begin{split}
            f_{\mu,7}(x) =&\hspace{5pt} (01)^{\frac{a}{2}}0^{b-1}1\text{, and}\\
            f_{\mu,7}^2(x) =&\hspace{5pt} (01)^{\frac{a}{2}}(01)^{\frac{b}{2}}\text{.}\\
        \end{split}
    \end{equation*}

    Now, if we consider $x=0^{a_1}1^{b_1}0^{a_2}1^{b_2}$ such that $a_1+a_2+b_1+b_2=n$. We know that there must be an even number of odd sequences. 
\begin{enumerate}
\item If all are even, then:
    \begin{equation*}
        \begin{split}
            f_{\mu,7}(x) =&\hspace{5pt} (01)^{\frac{a_1}{2}}0^{b_1-1}1(01)^{\frac{a_2}{2}}0^{b_2-1}1\text{, and}\\
            f_{\mu,7}^2(x) =&\hspace{5pt} (01)^{\frac{a_1}{2}}(01)^{\frac{b_1}{2}}(01)^{\frac{a_2}{2}}(01)^{\frac{b_2}{2}}\text{.}\\
        \end{split}
    \end{equation*}
\item If $a_1,b_1$ are odd, and the rest are even, then:
    \begin{equation*}
        \begin{split}
            f_{\mu,7}(x) =&\hspace{5pt} (01)^{\frac{a_1-1}{2}}0^{b_1-1}(01)^{\frac{a_2+2}{2}}0^{b_2-1}1\text{, and}\\
            f_{\mu,7}^2(x) =&\hspace{5pt} (01)^{\frac{a_1-1}{2}}(01)^{\frac{b_1-1}{2}}(01)^{\frac{a_2+2}{2}}(01)^{\frac{b_2}{2}}\text{.}\\
        \end{split}
    \end{equation*}
    \item If $a_1,a_2$ are odd, and the rest are even, then:
    \begin{equation*}
        \begin{split}
            f_{\mu,7}(x) =&\hspace{5pt} (01)^{\frac{a_1-1}{2}}0^{b_1}1(01)^{\frac{a_2-1}{2}}0^{b_2}1\text{, and}\\
            f_{\mu,7}^2(x) =&\hspace{5pt} (01)^{\frac{a_1-1}{2}}(01)^{\frac{b_1}{2}}(01)^{\frac{a_2+1}{2}}(01)^{\frac{b_2}{2}}\text{.}\\
        \end{split}
    \end{equation*}
    \item If $a_1,b_2$ are odd, and the rest are even, then:
    \begin{equation*}
        \begin{split}
            f_{\mu,7}(x) =&\hspace{5pt} (01)^{\frac{a_1-1}{2}}0^{b_1}1(01)^{\frac{a_2}{2}}0^{b_2-1}1\text{, and}\\
            f_{\mu,7}^2(x) =&\hspace{5pt} (01)^{\frac{a_1-1}{2}}(01)^{\frac{b_1}{2}}(01)^{\frac{a_2}{2}}(01)^{\frac{b_2+1}{2}}\text{.}\\
        \end{split}
    \end{equation*}
    \item If $b_1,a_2$ are odd, and the rest are even, then:
    \begin{equation*}
        \begin{split}
            f_{\mu,7}(x) =&\hspace{5pt} (01)^{\frac{a_1}{2}}0^{b_1-1}1(01)^{\frac{a_2-1}{2}}0^{b_2-1}1\text{, and}\\
            f_{\mu,7}^2(x) =&\hspace{5pt} (01)^{\frac{a_1}{2}}(01)^{\frac{b_1+1}{2}}(01)^{\frac{a_2}{2}}(01)^{\frac{b_2-1}{2}}\text{.}\\
        \end{split}
    \end{equation*}
    \item If $b_1,b_2$ are odd, and the rest are even, then:
    \begin{equation*}
        \begin{split}
            f_{\mu,7}(x) =&\hspace{5pt} (01)^{\frac{a_1}{2}}0^{b_1-1}1(01)^{\frac{a_2}{2}}0^{b_2-1}1\text{, and}\\
            f_{\mu,7}^2(x) =&\hspace{5pt} (01)^{\frac{a_1}{2}}(01)^{\frac{b_1+1}{2}}(01)^{\frac{a_2}{2}}(01)^{\frac{b_2-1}{2}}\text{.}\\
        \end{split}
    \end{equation*}
    \item If $a_2,b_2$ are odd, and the rest are even, then:
    \begin{equation*}
        \begin{split}
            f_{\mu,7}(x) =&\hspace{5pt} (01)^{\frac{a_1}{2}}0^{b_1-1}1(01)^{\frac{a_2-1}{2}}0^{b_2}1\text{, and}\\
            f_{\mu,7}^2(x) =&\hspace{5pt} (01)^{\frac{a_1}{2}}(01)^{\frac{b_1}{2}}(01)^{\frac{a_2+b_2}{2}}\text{.}\\
        \end{split}
    \end{equation*}
\item If all are odd, then:
    \begin{equation*}
        \begin{split}
            f_{\mu,7}(x) =&\hspace{5pt} (01)^{\frac{a_1-1}{2}}0^{b_1}1(01)^{\frac{a_2-1}{2}}0^{b_2}1\text{, and}\\
            f_{\mu,7}^2(x) =&\hspace{5pt} (01)^{\frac{a_1+b_1}{2}}(01)^{\frac{a_2+b_2}{2}}\text{.}\\
        \end{split}
    \end{equation*}
\end{enumerate}
Now, let $x=0^{a_1}1^{b_1}\dots0^{a_k}1^{b_k}$ with $\sum_{i=1}^k a_i+b_i=n$.
After the first step of $f_{\mu,7}$, all isles of $0$s of even number will become sequences of $(01)$ of half the length of the original isle, while the ones of odd length will become sequences of $(01)$ of half the original length minus $1$. At the same time, all isles of $1$s will become isles of $0$s whose length will depend on the parity of the isle of $0$s to its left, followed by one $1$ at the end. 

In the next step, as we saw in the case with two isles of $1$s and two of $0$s, the isles of $0$s belonging to $f_{\mu,7}(x)$ become sequences of $(01)$ that follow the order established by the first isle of $0$s of the original configuration. From where we obtain that $f_{\mu,7}^2(0^{a_1}1^{b_1}\dots0^{a_k}1^{b_k})=(01)^{\frac{n}{2}}$.

Let us start with $x=0^{a_1}1^b0^{a_2}$, since $f_7(000)=1$, the first cell becomes $1$, and thus 
$f_{\mu,7}(x)=(10)^{\frac{a_1-1}{2}}10^{b-1}(10)^{\frac{a_2+1}{2}}$ if $a_1,a_2$ are odd, 
and $f_{\mu,7}(x)=(10)^{\frac{a_1}{2}}0^{b-1}(10)^{\frac{a_2+1}{2}}1$ if $a_2,b$ are odd, 
$f_{\mu,7}(x)=(10)^{\frac{a_1-1}{2}}10^{b-1}1(01)^{\frac{a_2}{2}}1$ if $a_1,b$ are odd. If $a_2$ is odd (first two subcases), the result is $f_{\mu,7}(x)=(10)^{\frac{n}{2}}$, if $a_2$ is even, it is $f_{\mu,7}(x)=(01)^{\frac{n}{2}}$.

In the case with  $x=0^{a_1}1^{b_1}\dots0^{a_k}1^{b_k}0^{a_{k+1}}$ with $\sum_{i=1}^k (a_i+b_i)+a_{k+1}=n$, the first sequence will be $10$ instead of $01$, and then the isles of $1$s become isles of $0$s, from where all configurations where $a_{k+1}$ is odd result in $f_{\mu,7}(x)=(10)^{\frac{n}{2}}$, while if $a_{k+1}$ is even, then the result is $f_{\mu,7}(x)=(01)^{\frac{n}{2}}$.

It is easy to see that if we start with $x'=1^{b_1'}0^{a_1'}\dots1^{b_k'}0^{a_k'}1^{b_{k+1}'}$ with $\sum_{i=1}^k (a_i'+b_i') +b_{k+1}'=n$, then $f_{\mu,7}(x')$ will start with an isle of $0$s and it will end in a $1$, that corresponds to the case  $x=0^{a_1}1^{b_1}\dots0^{a_k}1^{b_k}$ with $\sum_{i=1}^k a_i+b_i=n$, which we have already studied.

    Therefore, any initial configuration can reach a fixed point under $f_{\mu,7}$.

    The analysis for Rule 15 is very similar, except that since $f_{15}(011)=1$, we are able to reach the fixed point after just one step.
\end{proof}
\begin{thm}
    For Rules $23$ and $30$, if the size of the ring is even, there exists an $E$-covering subset of sequential update modes.
\end{thm}
\begin{proof}
By definition of Rules $23$ and $30$, these are the fixed points for a ring of size $n$, if and only if $n$ is even: $(01)^{\frac{n}{2}}$ and $(10)^{\frac{n}{2}}$. Additionally, $0^n$ is always a fixed point for Rule $30$ ($f_{30}(000)=0$). Now, let $x=0^{a_1}1^{b_1}\dots 0^{a_k}1^{b_k}$, with $a_1,b_1,\dots,a_k,b_k\in\{0,\dots,n\}$ such that $\sum_{i=1}^k a_i+b_i=n$ an arbitrary initial configuration.

In this proof, a sequential update mode which leads $x$ to a fixed point is constructed.

In the case of Rule $23$, we start by updating the first cell of each isle of $1$s if it is necessary to ensure that the number of cells between the first $1$ of each isle is odd. This would increase the size of the corresponding isle of $0$s, since $f_{23}(011)=0$. Let $L$ be the set of cells that need to be updated.
The resulting configuration is $x'=x^{0.|L|}=0^{a_1'}1^{b_1'}\dots 0^{a_k'}1^{b_k'}$, with $a_1',b_1',\dots,a_k',b_k'\in\{0,\dots,n\}$ such that $\sum_{i=1}^k a_i'+b_i'=n$ and $a_{i+1}'+b_i'\equiv 0\mod2$, for all $i\in\{1,\dots,k-1\}$.

Then, we must update every other cell of the isles of $1$s (using that $f_{23}(111)=0$), which leaves only isolated $1$s. Note that if the isle if of even size, we would have to update the last cell of the isle (since $f_{23}(110)=0$). The result is as follows:
\begin{equation*}
\begin{split}
x^{0.|L|+1}&=0^{a_1'}0101^{b_1'-2}\dots 0^{a_k'}1^{b_k'}\text{,}\\
x^{0.|L|+2}&=0^{a_1'}(0101)01^{b_1'-4}\dots 0^{a_k'}1^{b_k'}\text{,}\\
x^{0.|L|+3}&=0^{a_1'}(01)^301^{b_1'-6}\dots 0^{a_k'}1^{b_k'}\text{,}\\
\vdots&\\
x^{0.|L|+\left\lfloor\frac{b_1'}{2}\right\rfloor}&=0^{a_1'}(01)^{\left\lfloor\frac{b_1'}{2}\right\rfloor}0^{a_2'}1^{b_2'}\dots 0^{a_k'}1^{b_k'}\text{.}\\
\end{split}
\end{equation*}
We continue as such with the rest of the isles of $1$s, after which we have:
 $$x''=x^{0.|L|+\sum_{i=1}^k\left\lfloor\frac{b_i'}{2}\right\rfloor}=0^{a_1'}(01)^{\left\lfloor\frac{b_1'}{2}\right\rfloor}0^{a_2'}(01)^{\left\lfloor\frac{b_2'}{2}\right\rfloor}\dots 0^{a_k'}(01)^{\left\lfloor\frac{b_k'}{2}\right\rfloor}\text{.}$$

Afterwards, we update every other cell of the isles of $0$s, starting with the first one:

\begin{equation*}
\begin{split}
x^{0.|L|+\sum_{i=1}^k\left\lfloor\frac{b_i'}{2}\right\rfloor+1}&=010^{a_1'-2}0101^{b_1'-2}\dots 0^{a_k'}1^{b_k'}\text{,}\\
x^{0.|L|+\sum_{i=1}^k\left\lfloor\frac{b_i'}{2}\right\rfloor+2}&=(01)^20^{a_1'-4}(0101)01^{b_1'-4}\dots 0^{a_k'}1^{b_k'}\text{,}\\
x^{0.|L|+\sum_{i=1}^k\left\lfloor\frac{b_i'}{2}\right\rfloor+3}&=(01)^30^{a_1'}(01)^301^{b_1'-6}\dots 0^{a_k'}1^{b_k'}\text{,}\\
\vdots&\\
x^{0.|L|+\sum_{i=1}^k\left\lfloor\frac{b_i'}{2}\right\rfloor+\left\lfloor\frac{a_1'}{2}\right\rfloor}&=(01)^{\left\lfloor\frac{a_1'}{2}\right\rfloor}(01)^{\left\lfloor\frac{b_1'}{2}\right\rfloor}0^{a_2'}1^{b_2'}\dots 0^{a_k'}1^{b_k'}\text{.}\\
\end{split}
\end{equation*}

Continuing as such we obtain a configuration $x'''=(01)^{\frac{n}{2}}$. Finally, since we have already reached the fixed point, we proceed to update the cells that have not been updated yet and since $f_{23}(010)=1$ and $f_{23}(101)=0$, their respective state cannot change, and thus we have found a sequential update mode such that $f_{\mu,23}(x)=(01)^{\frac{n}{2}}$. 

Note that if the initial configuration is instead $y=1^{a_1}0^{b_1}\dots 1^{a_k}0^{b_k}$, with $a_1,b_1,\dots,a_k,b_k\in\{0,\dots,n\}$ such that $\sum_{i=1}^ka_i+b_i=n$, we can instead obtain $f_{\mu,23}(y)=(10)^{\frac{n}{2}}$ following a similar reasoning.\bigskip

In the case of Rule $30$, we update the last cell of each isle of $0$s, using that $f_{30}(001)=1$, if it is needed to ensure that the number of cells between the first $1$ of each isle is odd.
Let $L$ be the set of cells that we need to update.

Let $x'=x^{0.|L|}=0^{a_1}1^{b_1}0^{a_2'}1^{b_2'}\dots0^{a_k'}1^{b_k'}$, with $a_2',b_2',\dots,a_k',b_k'$ the new size of the isles, where $a_i'=a_i-1$ and $b_i'=b_i+1$, if the change was necessary.

Then, we update the cells of the first isle of $0$s from left to right, using that $f_{30}(100)=1$ to turn the isles of $0$s into isles of $1$s:
\begin{equation*}
\begin{split}
x^{0.|L|}&=0^{a_1}1^{b_1}0^{a_2'}1^{b_2'}\dots0^{a_k'}1^{b_k'}\text{,}\\
x^{0.|L|+1}&=10^{a_1-1}1^{b_1}0^{a_2'}1^{b_2'}\dots0^{a_k'}1^{b_k'}\text{,}\\
x^{0.|L|+2}&=1^20^{a_1-2}1^{b_1}0^{a_2'}1^{b_2'}\dots0^{a_k'}1^{b_k'}\text{,}\\
\vdots&\\
x^{0.|L|+a_1-1}&=1^{a_1-1}0^{1}1^{b_1}0^{a_2'}1^{b_2'}\dots0^{a_k'}1^{b_k'}\text{,}\\
\end{split}
\end{equation*}
and since $f_{30}(101)=0$, the last $0$ of the first isle does not change and thus we have 
$$x^{0.|L|+a_1}=1^{a_1-1}0^{1}1^{b_1}0^{a_2'}1^{b_2'}\dots0^{a_k'}1^{b_k'}\text{.}$$

Then, we continue with each isle of $0$s, to create a configuration such as $$x''=x^{0.\sum_{i=1}^k a_i}=1^{a_1-1}01^{b_1}\dots1^{a_k'-1}01^{b_k'}\text{.}$$ 
Note that we now have updated all cells whose state was originally $0$. 

Next, we update the cells corresponding to the original first isle of $1$s from left to right, using that $f_{30}(011)=1$ and $f_{30}(111)=0$, which creates the sequence $(10)^{\left\lfloor\frac{b_1}{2}\right\rfloor}$, as follows:
\begin{equation*}
\begin{split}
x^{0.1+\sum_{i=1}^k (a_i)}&=1^{a_1-1}01^{b_1}\dots1^{a_k'-1}01^{b_k'}\text{,}\\
x^{0.2+\sum_{i=1}^k (a_i)}&=1^{a_1-1}0101^{b_1-1}\dots1^{a_k'-1}01^{b_k'}\text{,}\\
\vdots&\\
x^{0.b_1+\sum_{i=1}^k (a_i)}&=1^{a_1-1}0(10)^{\left\lfloor\frac{b_1}{2}\right\rfloor}\dots1^{a_k'-1}01^{b_k'}\text{.}\\
\end{split}
\end{equation*} 
Then we can continue to update from left-to-right each isle of $1$s of the original configuration, after which we obtain
$$f_{\mu,30}(x)=x^1=1^{a_1-1}0(10)^{\left\lfloor\frac{b_1}{2}\right\rfloor}\dots1^{a_k'-1}0(10)^{\frac{b_k'}{2}}\text{.}$$
At this point we have updated all the cells whose original state was $1$, and therefore, we have updated all the cells of the ring.
 
If we apply one step of the rule once again, following the update mode that we have previously established, we obtain $f_{\mu,30}^{2}(x)=(01)^{\frac{n}{2}}$.  

Analogously, if the initial configuration is  $y=1^{a_1}0^{b_1}\dots 0^{a_k}1^{b_k}$, with $a_1,b_1,\dots,a_k,b_k\in{0,\dots,n}$ such that $\sum_{i=1}^ka_i+b_i=n$, we can obtain $f_{\mu,30}^2(y)=(10)^{\frac{n}{2}}$ following a similar reasoning.

Thus, we have found a subset of sequential update modes that ensure that every configuration can reach a fixed point.
\end{proof}

\subsubsection{Restriction to Rings of Size Multiple of $3$}

Rule $45$ is quite close to the three previous rules in the sense that it is possible to make it converge only under specific conditions. Indeed, for Rule $45$ to admit fixed points, Lemma~\ref{lem:45} below, states that the size of the grid must be a multiple of $3$. Thus, let $T\subsetneq\Bn$ the set of configuration whose size is a multiple of $3$. {Note that the convergence of Rule $45$ (and $37$) over rings whose size is a multiple of $3$ has been studied in~\cite{Sethi2018size} for the asynchronous case. Here we will show what occurs under periodic update modes for Rule $45$, along with conjectures for sequential update modes for both Rules $37$ and $45$. Thus, in this section we will focus on $T$-coverings for Rule $45$.}

\begin{table}[t]
\begin{center}
\begin{tabular}{|c|c|c|c|c|c|c|c|c|}
    \hline
         &111&110&101&100&011&010&001&000  \\\hline
         45&0&0&1&0&1&1&0&1\\\hline 
    \end{tabular}    
    \caption{Definition of Rule $45$.}
    \label{tab:rule45}
\end{center}
\end{table}
\begin{lemma}\label{lem:45}
For Rule $45$ there are three fixed points for a ring of size $n$, if and only if $n$ is a multiple of $3$: $(001)^{\frac{n}{3}}$, $(010)^{\frac{n}{3}}$ and $(100)^{\frac{n}{3}}$.
\end{lemma}
\begin{proof}
    It is easy to see that $(001)^{\frac{n}{3}}$, $(010)^{\frac{n}{3}}$ and $(100)^{\frac{n}{3}}$ are fixed points, given that $f_{45}(010)=1$, and $f_{45}(001)=f_{45}(100)=0$ (see Table~\ref{tab:rule45}). 
\end{proof}

For the following proof we need to introduce the concept of \textit{temporal composition}. If we have two different update modes $\mu_1$ and $\mu_2$ and apply a Rule $f$ under the first one and then under the second one, the resulting dynamical system will be a temporal composition of the two update modes,~\ie,  $f_{\mu}=f_{\mu_2}\circ f_{\mu_1}$. It is easy to see that if we combine two sequential update modes the resulting one will be a periodic update mode that is not sequential.

\begin{thm}
    Any initial configuration of size $n$, with $n$ multiple of 3, can converge to a fixed point under a periodic update mode for Rule $45$.
\end{thm}

\begin{proof}
   By Lemma~\ref{lem:45} we know that $(001)^{\frac{n}{3}}$, $(010)^{\frac{n}{3}}$ and $(100)^{\frac{n}{3}}$. are the only fixed points of Rule $45$. Furthermore, since $f_{45}(111)=f_{45}(110)=0$ and $f_{45}(000)=1$, we have that isles of two or more $1$s and isles of three or more $0$s cannot exist in a fixed point. 
    
Let us perform a case by case analysis.
    \begin{enumerate}
        \item Let $x=0^n$. We can define a sequential update mode $$\mu_0 = (0,3,6,\dots,n-3,1,4,7, \dots,n-2,2,5,8,\dots,n-1)$$ which results in $f_{45,\mu_0}(0^n)=(100)^{\frac{n}{3}}$.    
        \item Let $x=1^n$. We can define a sequential update mode $$\mu_1 = (2,5,8,\dots,n-1,1,4,7, \dots,n-2,0,3,6,\dots,n-3)$$ which results in $f_{\mu_1,45}(1^n)=(100)^{\frac{n}{3}}$.
        \item Let $x=0^{a}1^{b}$, such that $a+b=n$.
        \begin{enumerate}
            \item If $a\mod 3\equiv 0$ then $b\mod 3\equiv 0$, thus we can define a sequential update mode
            \begin{equation*}
                \begin{split}
                    \mu_{2}=(&a+2,a+5,a+8,\dots,n-1,\\
                    &a+1,a+4,a+7,\dots,n-2,\\&a,a+3,a+6,\dots,n-3,\\&0,3,6,\dots,a-3,\\&1,4,7,\dots,a-2,\\&2,5,8,\dots,a-1)
                \end{split}
            \end{equation*}
            which results in $f_{\mu_2,45}(0^{a}1^{b})=(100)^{\frac{n}{3}}$.
            
            \item If $a\mod 3\equiv 1$ then $b\mod 3\equiv 2$, thus we can define  a sequential update mode
            \begin{equation*}
                \begin{split}
                    \mu_{3}=(&a+2,a+5,a+8,\dots,n-3,\\& a+1,a+4,a+7,\dots,n-1,\\&a,a+3,a+6,\dots,n-2,\\&1,4,7,\dots,a-3,\\&0,3,6,\dots,a-1,\\&2,5,8,\dots,a-2)
                \end{split}
            \end{equation*}
            which results in $f_{45,\mu_3}(0^{a}1^{b})=(010)^{\frac{n}{3}}$.
            
            \item If $a\mod 3\equiv 2$ then $b\mod 3\equiv 1$, thus we can define  a sequential update mode
            \begin{equation*}
                \begin{split}
                    \mu_{4}=(&a+2,a+5,a+8,\dots,n-2,\\&a+1,a+4,a+7,\dots,n-3,\\&a,a+3,a+6,\dots,n-1,\\&2,5,8,\dots,a-3,\\&1,4,7,\dots,a-1,\\&0,3,6,\dots,a-2)
                \end{split}
            \end{equation*}
            which results in $f_{\mu_4,45}(0^{a}1^{b})=(001)^{\frac{n}{3}}$.
        \end{enumerate}
        The procedure is analogous for $x=1^b0^{a}$.
        \item Let $x=0^{a_1}1^{b_1}0^{a_2}1^{b_2}$. Let $\ha = a_1+b_1$.
        We need to ensure the correct distance between the two isles of $1$s. 
        \begin{enumerate}
            \item If  $(b_1+a_2)\mod 3\equiv0$  we can simply update the isles of $1$s from right-to-left, and then we can update the original isles of $0$s every third cell, according to a convenient distance between the remaining $1$s, following the same logic as with case 3.
            \begin{equation*}
                \begin{split}
                    \mu_5=(&n-1,n-2,\dots,n-b_2,\\&\ha-1,\ha-2,\dots,a_1,\ha+a_2-3,\\
                          &\ha+a_2-6,\dots,\ha,\\&\ha+a_2-2,\ha+a_2-5,\dots,\ha+a_2-1,\\
                          &\ha+a_2-4,\dots,a_1-3,a_1-6,\dots,0,\\
                          &a_1-2,a_1-5,\dots,a_1-1,a_1-4)\text{.}
                \end{split}
            \end{equation*}
            Note that if $a_1\mod 3\equiv2$, we must update the first cell of the configuration to prevent an unwanted $1$ to appear.

            The resulting configuration will be 
            $$f_{\mu_5,45}(x)=(100)^{\left\lfloor\frac{a_1}{3}\right\rfloor}10^{b_1-1}(100)^{\left\lfloor\frac{a_2}{3}\right\rfloor}10^{b_2-1}\text{,}$$
            if $a\mod 3\equiv0$, or
            $$f_{\mu_5,45}(x)=(010)^{\left\lfloor\frac{a_1}{3}\right\rfloor}010^{b_1-1}(010)^{\left\lfloor\frac{a_2}{3}\right\rfloor}010^{b_2-1}\text{,}$$
            if $a_1\mod 3\equiv1$, or in the especial case $a_1\mod 3\equiv2$
            $$f_{\mu_5,45}(x)=(001)^{\left\lfloor\frac{a_1}{3}\right\rfloor}0010^{b_1-1}(001)^{\left\lfloor\frac{a_2}{3}\right\rfloor}0010^{b_2-1}\text{.}$$
            
            \item If $(b_1+a_2)\mod 3\equiv2$, we have to start with cell $\ha+a_2-2$, then we must update cell $\ha+a_2-1$ to correct the distance between the isles of $1$s, and then we proceed with the same order as the one established in Case 4(a).
            \begin{equation*}
                \begin{split}
                    \mu_6=(&\ha+a_2-2,\ha+a_2-1,\\
 &n-1,n-2,\dots,n-b_2,\\
 &\ha-1,\ha-2,\dots,a_1,\\
 &\ha+a_2-3,\ha+a_2-6,\dots,\ha,\\
  &\ha+a_2-5,\dots,\ha+a_2-4,\dots,a_1-3,\\                          &a_1-6,\dots,0,\\
 &a_1-2,a_1-5,\dots,a_1-1,a_1-4)\text{.}
                \end{split}
            \end{equation*}

            The resulting configuration will be 
            $$f_{\mu_6,45}(x)=(100)^{\left\lfloor\frac{a_1}{3}\right\rfloor}10^{b_1-1}(100)^{\left\lfloor\frac{a_2-2}{3}\right\rfloor}110^{b_2}\text{,}$$
            if $a\mod 3\equiv0$, or
            $$f_{\mu_6,45}(x)=(010)^{\left\lfloor\frac{a_1}{3}\right\rfloor}010^{b_1-1}(010)^{\left\lfloor\frac{a_2-2}{3}\right\rfloor}0110^{b_2}\text{,}$$
            if $a_1\mod 3\equiv1$, or in the especial case $a_1\mod 3\equiv2$
            $$f_{\mu_6,45}(x)=(001)^{\left\lfloor\frac{a_1}{3}\right\rfloor}0010^{b_1-1}(001)^{\left\lfloor\frac{a_2-2}{3}\right\rfloor}00110^{b_2}\text{.}$$
            
            \item If $(b_1+a_2)\mod 3\equiv1$ we have to start with cell $a_1-2$, then we must update cell $a_1-1$ to correct the distance between the isles of $1$s, and then we proceed with the same order as the one established in Case 4(a).
            \begin{equation*}
                \begin{split}
                    \mu_7=(&a_1-2,a_1-1,\\
&n-1,n-2,\dots,n-b_2,\ha-1,\ha-2,\dots,a_1,\\
&\ha+a_2-3,\ha+a_2-6,\dots,\ha,\\\
&\ha+a_2-2,\ha+a_2-5,\dots,\ha+a_2-1,\\
&\ha+a_2-4,\dots,a_1-3,\\
&a_1-6,\dots,0,a_1-5,\dots,a_1-4,a_1-7,\dots)\text{.}
                \end{split}
            \end{equation*}

            The resulting configuration will be 
            $$f_{\mu_7,45}(x)=(100)^{\left\lfloor\frac{a_1-2}{3}\right\rfloor}110^{b_1}(100)^{\left\lfloor\frac{a_2}{3}\right\rfloor}10^{b_2-1}\text{,}$$
            if $(a_1-2)\mod 3\equiv0$, or
            $$f_{\mu_7,45}(x)=(010)^{\left\lfloor\frac{a_1-2}{3}\right\rfloor}0110^{b_1}(010)^{\left\lfloor\frac{a_2}{3}\right\rfloor}010^{b_2-1}\text{,}$$
            if $(a_1-2)\mod 3\equiv1$, or in the especial case $(a_1-2)\mod 3\equiv2$
            $$f_{\mu_7,45}(x)=(001)^{\left\lfloor\frac{a_1-2}{3}\right\rfloor}00110^{b_1}(001)^{\left\lfloor\frac{a_2}{3}\right\rfloor}0010^{b_2-1}\text{.}$$
        \end{enumerate}
        
        Then we create a new sequential update mode following the reasoning established in Case 1 to work with the two longer isles of $0$s.

        Note that this results in a temporal combination of two different sequential update modes, which results in a periodic update mode.
    \end{enumerate}
    Extending case 4 to $x = 0^{a_1}1^{b_1}\dots 0^{a_k}1^{b_k}$ and $x = 0^{a_1}1^{b_1}\dots 0^{a_k}1^{b_k}0^{a_{k+1}}$ is direct and analogous for $x = 1^{b_1}0^{a_1}\dots 1^{b_k}0^{a_k}$ and $x = 1^{b_1}0^{a_1}\dots 1^{b_k}0^{a_k}1^{b_{k+1}}$.
\end{proof}

Our result is not satisfactory in the context of our study, given that the update mode we have obtained is periodic but not sequential. Therefore, in order to find candidates for sequential update modes under which Rule $45$ might be able to lead all initial configurations to a fixed point, we ran computer simulations for different sized rings. These simulations considered all possible sequential update modes for rings of size $6$ and $9$. As a result we found that the best candidate to prove the existence of a universal sequential update mode is the one presented in the following conjecture.

\begin{conjecture}\label{cnjtr:rule45}
    Given a ring whose size $n$ is a multiple of $3$, $\mu=(0,1,2,\dots,n-2,n-1)$ is a $T$-universal update mode for Rule 45. 
\end{conjecture}

\begin{remark}
    Unlike for Rules 7 and 15, the simulations suggest that the update mode in this case is not unique. Indeed, there are 117 sequential update modes for $n=9$ that lead all configurations to a fixed point, and 15 sequential update modes for $n=6$. However, the formal analysis involves a great amount of sub-cases and could not be completed in a reasonable amount of time.
\end{remark}

\subsection{Without Covering of Sequential Update Modes}
In this section we pay attention to the rules for which, even though they have fixed points, there are configurations that can never reach them, regardless of the sequential update modes being used.

\begin{table}[t]
\begin{center}
\begin{tabular}{|c|c|c|c|c|c|c|c|c|}
    \hline
         &111&110&101&100&011&010&001&000  \\\hline
         38&0&0&1&0&0&1&1&0\\\hline 
         46&0&0&1&0&1&1&1&0\\\hline 
         54&0&0&1&1&0&1&1&0\\\hline 
         60&0&0&1&1&1&1&0&0\\\hline 
         62&0&0&1&1&1&1&1&0\\\hline 
    \end{tabular}    
    \caption{Definition of Rule $38$, $46$, $54$, $60$ and $62$.}
    \label{tab:isolated0}
\end{center}
\end{table}

{In~\cite{Fates2018} it is proven that the homogeneous configuration $0^n$ is not reachable under asynchronous update modes for Rules $38$, $46$, $54$ and $60$. Here, we present an alternative way of proving that result, which, under sequential update modes, can add Rule $62$ to the list.}

\begin{thm}\label{thm:isolated0}
    For Rules 38, 46, 54, 60 and 62, configuration $0^n$ is the only fixed point and it is isolated under sequential update modes.
\end{thm}
\begin{proof}
	Note that, for all these rules, $0^n$ is the only fixed point in parallel, therefore, it is the only fixed point under sequential update modes.
	
    We will focus on Rule $38$, the other rules follow a similar reasoning. Let $\mu$ be any sequential update mode. 
    
    Let us consider an initial configuration $x=0^{a_1}1^{b_1}\dots0^{a_k}1^{b_k}$ such that $\sum_{i=1}^k a_i+b_i=n$. Given that $f_{38}(011)=f_{38}(110)=0$ (see Table~\ref{tab:isolated0}), we can always reduce the size of the isles of $1$s until all we are left with are isolated $1$s. Let $\ell_0$ be a cell that contains one of the isolated $1$s. 
    
    Thanks to the definition of Rule $38$, we know that the isolated $1$s we have left must have been updated after the cells on either side of them (otherwise its state would have to be $0$, since $f_{38}(011)=f_{38}(110)=f_{38}(111)=0$). Let us notate this by saying that $\mu_{\ell_0-1}<\mu_{\ell_0}$ and  $\mu_{\ell_0}>\mu_{\ell_0+1}$.

    Because $f_{38}(001)=1$, we know that the cells to the left of each $1$ will change its state to $1$, until the growing isle of $1$s reaches a cell $\ell_1$ such that $\ell_1-1$ is updated before $\ell_1$, that is: $\mu_{\ell_1-1}<\mu_{\ell_1}$. Additionally, the state of cell $\ell_0$ will now be $0$, because $f(110)=0$.

    Since the previous statement holds for all isolated $1$s, that is, all isolated $1$s will create a growing isles of $1$s that start on a cell $\ell_i$ and finish on $\ell_{i-1}-1$, with $i\in\{0,\dots,k'\}$, with $k'$ the number of isles. Note that the isles of $1$s could merge, since $f_{38}(101)=1$, thus the new number of isles $k'$ could be different from $k$, the number of isles of the original configuration.

    Thus, once we reach isolated $1$s, in the next step the number of $1$s must increase (by increasing the size of the isles) or stay constant (if each $1$ moved exactly one cell to the left). Therefore, configuration $0^n$ cannot be reached.
\end{proof}

\begin{table}[t]
\begin{center}
\begin{tabular}{|c|c|c|c|c|c|c|c|c|}
    \hline
         &111&110&101&100&011&010&001&000  \\\hline
         110&0&1&1&0&1&1&1&0\\\hline 
         126&0&1&1&1&1&1&1&0\\\hline 
    \end{tabular}    
    \caption{Definition of Rule $110$ and $126$.}
    \label{tab:isolated0_2}
\end{center}
\end{table}
\begin{thm}\label{thm:isolated0_2}
    For Rules 110 and 126, configuration $0^n$ is the only fixed point and it is isolated under sequential update modes.
\end{thm}
\begin{proof}
	Note that $0^n$ is the only fixed point in parallel, and thus, it is the only possible fixed point under sequential update modes.
	
	Let $\mu$ be a sequential update mode.
    Unlike Theorem~\ref{thm:isolated0}, for these rules we have that $f(110)=f(011)=1$ (see Table~\ref{tab:isolated0_2}), which means that we cannot simply try to reduce the size of the isles of $1$s. Instead, the best that can happen is that we use $f(111)=0$ to break them up, leaving $1$s alone or in pairs.

    But once again, since $f(010)=1$, those $1$s cannot disappear and instead, since $f(001)=f(101)=1$, the number of $1$s will not decrease, therefore $0^n$ cannot be reached.
\end{proof}

\begin{table}[t]
\begin{center}
\begin{tabular}{|c|c|c|c|c|c|c|c|c|}
    \hline
         &111&110&101&100&011&010&001&000  \\\hline
         134&1&0&0&0&0&1&1&0\\\hline 
         142&1&0&0&0&1&1&1&0\\\hline 
         150&1&0&0&1&0&1&1&0\\\hline 
         156&1&0&0&1&1&1&0&0\\\hline 
    \end{tabular}    
    \caption{Definition of Rule $134$, $142$, $150$ and $156$.}
    \label{tab:isolated0_and_1}
\end{center}
\end{table}
\begin{thm}\label{thm:isolated0_and_1}
    Rules $134$, $142$, $150$ and $156$ only have isolated fixed points under sequential update modes:
    \begin{itemize}
        \item $0^n$ and $1^n$ if $n$ is odd.
        \item $0^n$, $1^n$, $(01)^{\frac{n}{2}}$ and $(10)^{\frac{n}{2}}$ if $n$ is even.
    \end{itemize}
\end{thm}

\begin{proof}
    We know that $0^n$ and $1^n$ are fixed points, and since $f(101)=0$ and $f(010)=1$ (see Table~\ref{tab:isolated0_and_1}) we know that the isolated $0$s and $1$s do not disappear for any of these rules, meaning that said fixed points cannot be reached.

    Now, let $n$ be even and $x\not\in\{0^n,1^n,(01)^{\frac{n}{2}},(10)^{\frac{n}{2}}\}$. 
    
    Given that $f(000)=0$ and $f(111)=1$, we know that new isles cannot be created. This means that the number of isles cannot increase which implies that the only way to have $\frac{n}{2}$ isolated $1$s (or $0$s) is for them to have existed in the initial configuration.    

    Therefore, all four possible fixed points are isolated ones.
\end{proof}

\begin{table}[t]
\begin{center}
\begin{tabular}{|c|c|c|c|c|c|c|c|c|}
    \hline
         &111&110&101&100&011&010&001&000  \\\hline
         105&0&1&1&0&1&0&0&1\\\hline 
    \end{tabular}    
    \caption{Definition of Rules $105$.}
    \label{tab:105}
\end{center}
\end{table}
\begin{thm}\label{thm:105seq}
    Rule 105 under sequential update modes only has fixed points if the length of the ring is a multiple of 4, and said fixed points are isolated: $(0011)^{\frac{n}{4}}$, $(0110)^{\frac{n}{4}}$, $(1001)^{\frac{n}{4}}$ and $(1100)^{\frac{n}{4}}$.
\end{thm}

\begin{proof}
    It is easy to see that $x=(0011)^{\frac{n}{4}}$, $x=(0110)^{\frac{n}{4}}$, $x=(1001)^{\frac{n}{4}}$ and $x=(1100)^{\frac{n}{4}}$ with $n\mod 4\equiv 0$ are fixed points, since $f_{105}(011)=f_{105}(110)=1$ and $f_{105}(001)=f_{105}(100)=0$ (see Table~\ref{tab:105}). 
    
    Let $\mu$ be a sequential update mode.
    Note that if a configuration has isles of one $1$, said isles cannot contain two $1$s in the next substep, since $f_{105}(100)=f_{105}(001)=0$. Similarly, isolated $0$s cannot become pairs of $0$s, since $f_{105}(011)=f_{105}(110)=1$.

    Contrarily, if we have an initial configuration with isles of $1$s whose length is three or more, by definition these isles cannot decrease in size from the borders, instead, in order to obtain smaller isles we need to divide them ($f_{105}(000)=1$ and $f_{105}(111)=0$), which will create new borders that cannot be changed. However, $f_{105}(010)=0$ and $f_{105}(101)=1$, meaning that any isolated $1$s and $0$s that appear after one step can disappear after the next one.

    It follows that the only way of having pairs of $1$s followed by pairs of $0$s is for the initial configuration to contain them. Therefore, all four fixed points are isolated.
\end{proof}
\begin{corollary}
There does not exist a covering of sequential update modes for Rules $38$, $46$, $54$, $60$, $62$, $110$, $126$, $134$, $142$, $150$, $156$ and $105$.
\end{corollary}
\begin{proof}
It is direct from Theorems~\ref{thm:isolated0},~\ref{thm:isolated0_2},~\ref{thm:isolated0_and_1} and~\ref{thm:105seq}.
\end{proof}

So far, we have proven that there does not exist a covering by proving that the fixed points (when they exist) are isolated. In what follows we will instead prove that there are certain configurations that are unable to reach fixed points.

\begin{table}[t]
\begin{center}
\begin{tabular}{|c|c|c|c|c|c|c|c|c|}
    \hline
         &111&110&101&100&011&010&001&000  \\\hline
         73&0&1&0&0&1&0&0&1\\\hline 
    \end{tabular}    
    \caption{Definition of Rule $73$.}
    \label{tab:73}
\end{center}
\end{table}
\begin{thm}\label{thm:73}
    For Rule 73 under sequential update modes, a configuration with isolated $1$s can never reach a fixed point.
\end{thm}

\begin{proof}
Let $\mu$ be a sequential update mode.
    Note that configurations with three or more consecutive $1$s or isolated $1$s cannot be fixed points because $f_{73}(010)=f_{73}(111)=0$, meaning that a fixed point can only have $1$s in pairs. Note that this does not mean that any configuration with two consecutive $1$s surrounded by zeros is a fixed point.
    
    Similarly to Theorem~\ref{thm:105seq}, we have that $f_{73}(011)=f_{73}(110)=1$ and $f_{73}(001)=f_{73}(100)=0$, which we know results in that isles of $0$s and of $1$s cannot increase (nor decrease) in size, which includes isolated $1$s and $0$s.

    Therefore, configurations with isolated $1$s cannot contain in their orbit configurations with isles of $1$s of length two, which are the only ones that could be fixed points.
\end{proof}

\begin{table}[t]
\begin{center}
\begin{tabular}{|c|c|c|c|c|c|c|c|c|}
    \hline
         &111&110&101&100&011&010&001&000  \\\hline
         6 &0&0&0&0&0&1&1&0\\\hline 
         14&0&0&0&0&1&1&1&0\\\hline 
         22&0&0&0&1&0&1&1&0\\\hline 
    \end{tabular}    
    \caption{Definition of Rule $6$, $14$ and $22$.}
    \label{tab:6_14_22}
\end{center}
\end{table}

{As mentioned in Subsection~\ref{sssec:even}, in~\cite{Sethi2018size} it was shown that Rules $6$, $14$ and $22$ can always reach fixed points under asynchronous updating when the size of the ring is even. However, the following result does not contradict the one stated in said paper, since ours is specifically about sequential update modes, whose restrictions with respect to more general asynchronous update modes prevent certain configurations from reaching fixed points.}

\begin{thm}\label{thm:6_14_22}
    For Rules $6$, $14$ and $22$ under sequential update modes, a configuration with isolated $1$s can never reach a fixed point.
\end{thm}

\begin{proof}
    Let $x=0^{n-1}1$ be a configuration that only contains one cell whose state is $1$.
    
    Let $\mu$ be a sequential update mode.
    Since $f(010)=1$, we know that the $1$s cannot disappear. And given that $f(001)=1$, we can see that the number of $1$s will increase to the left until we reach a cell $\ell_1$ such that the cell to its left updates before cell $\ell_1$, which we will notate by $\mu_{\ell_1-1}<\mu_{\ell_1}$. 
    
    We know that each $\ell_1$ will be the first cell that updates before the one to its left, otherwise the isle would have stopped increasing on the previous cell with that characteristic.

    On the next step, the isles move to the next cell such that $\mu_{\ell_2-1}<\mu_{\ell_2}$. This means that we have an isle of $1$s moving from right to left, going from one $\ell_i$ to the next, resulting in cycles of length $|L|$, with $L=\{\ell \in \{0,\dots,n-1\}\mid \mu_{\ell-1}<\mu_{\ell}\}$.

    Note that since $f(101)=0$, if we have more than {a single} isolated $1$ at the start, the resulting isles cannot merge if they collide. This results in that the isle of $1$s to the right (the one that would advance into the position currently occupied) stops increasing before reaching an $\ell$ cell. But since the isles cannot disappear, this results in cycles of length $\mathcal{O}(|L|)$.

    Thus, configurations that contain isolated $1$s cannot reach a fixed point.
\end{proof}

\begin{table}[t]
\begin{center}
\begin{tabular}{|c|c|c|c|c|c|c|c|c|}
    \hline
         &111&110&101&100&011&010&001&000  \\\hline
         28&0&0&0&1&1&1&0&0\\\hline 
         29&0&0&0&1&1&1&0&1\\\hline 
    \end{tabular}    
    \caption{Definition of rules $28$ and $29$.}
    \label{tab:rules28_29}
\end{center}
\end{table}

\begin{thm}\label{thm:28_29}
    For Rules $28$ and $29$ under sequential update modes, a configuration that contains the word $01001$ can never reach a fixed point.
\end{thm}

\begin{proof}
    Note that, by definition (see Table~\ref{tab:rules28_29}), the word $w = 01$ is a wall for both rules, meaning that their value does not change regardless of the update mode and whatever the state of the configuration is around it.
	
	Let $\mu$ be a sequential update mode.
    Now, if we have a word $01001$, what we have is a single $0$ surrounded by walls. Since $f_{28}(100)=f_{29}(100)=1$ and $f_{28}(110)=f_{29}(110)=0$ this creates a cycle within the walls of length 2. Thus, any configuration that contains $01001$ is unable to reach a fixed point.
\end{proof}

\begin{table}[t]
\begin{center}
\begin{tabular}{|c|c|c|c|c|c|c|c|c|}
    \hline
         &111&110&101&100&011&010&001&000  \\\hline
         108&0&1&1&0&1&1&0&0\\\hline 
    \end{tabular}    
    \caption{Definition of Rule $108$.}
    \label{tab:108}
\end{center}
\end{table}
\begin{thm}\label{thm:108}
    For Rule $108$ under sequential update modes, a configuration can never reach a fixed point if it contains the word $0011100$.
\end{thm}

\begin{proof}
    Note that $100$ and $001$ are walls. Indeed, $f_{108}(000)=f_{108}(001)=f_{108}(100)=0$ and $f_{108}(011)=f_{108}(110)=f_{108}(010)=1$. Then, if we consider a configuration that contains the word $(001)1(100)$, what we have is a cell whose state is $1$ surrounded by two walls.

	Let $\mu$ be a sequential update mode.
    By definition, the cells where the walls are located cannot change their state and so, after one step that sequence will change to $(001)0(100)$ and after two it will return to its initial state. Therefore, any configuration containing $(001)1(100)$ and/or $(001)0(100)$ can never reach a fixed point.
\end{proof}

\begin{corollary}
There does not exist a covering of sequential update modes for Rules 73, 6, 14, 22, 28, 29 and 108.
\end{corollary}
\begin{proof}
It is a direct result from Theorems~\ref{thm:73},~\ref{thm:6_14_22},~\ref{thm:28_29} and~\ref{thm:108}.
\end{proof}
\subsubsection{Rule $37$}

\begin{table}[t]
\begin{center}
\begin{tabular}{|c|c|c|c|c|c|c|c|c|}
    \hline
         &111&110&101&100&011&010&001&000  \\\hline
         37&0&0&1&0&0&1&0&1\\\hline 
    \end{tabular}    
    \caption{Definition of Rule $37$.}
    \label{tab:37}
\end{center}
\end{table}
Rule $37$ is a special case. Like Rule $45$, it only has fixed points if the size of the ring is a multiple of $3$. Furthermore it has the same fixed point as Rule $45$: $(001)^{\frac{n}{3}}$, $(010)^{\frac{n}{3}}$ and $(100)^{\frac{n}{3}}$, which can easily be checked in Table~\ref{tab:37}.

However, unlike Rule $45$, the computer simulations we ran for Rule $37$ say that there is no $T$-universal sequential update mode. Moreover, simulations performed over rings of size $6$ and $9$, suggest that there does not even exist a $T$-covering sequential update mode. Indeed, for $n=6$ the configuration $x=001000$ was unable to reach a fixed point regardless of the sequential update mode, and for $n=9$ the same could be said for, among others, the configuration $x=000010001$. Thus we have the following conjecture:
\begin{conjecture}\label{cnjtr:rule37}
There does not exist a covering of sequential update modes for Rule $37$.
\end{conjecture}

\FloatBarrier
\section{Discussion}
\label{sec:discussion}

We have studied the convergence of all 88 {non-equivalent} elementary cellular automata, and have classified them into two separate groups: those that fit Theorem~\ref{thm:DAS} and the ones that fall outside of it.

Out of the 50 ECAs that fulfill one of the criteria given by  Theorem~\ref{thm:DAS} (presented in Table~\ref{tab:inside-theorem}), we have proven that there exists a universal sequential update mode for 13 rules, all of them belonging to class II under Wolfram's classification. 

\begin{table}[t]
\resizebox{\textwidth}{!}{
\begin{tabular}{|c|c|c|c|c|c|}
\hline
\diagbox[width=\dimexpr 24\tabcolsep\relax, height=1cm]{Classification}{Wolfram's class}&I&II&III&IV&Total\\\hline
{\makecell{All update modes\\are universal}}&\makecell{0,8,\\128,136}&\makecell{4,12,36,44,72,\\76,78,132,140,\\164,200,204}&$\emptyset$&$\emptyset$&16\\\hline
{\makecell{All sequential update\\modes are universal}}&\makecell{32,40,\\160,168}&\makecell{5,13,56,77,94,\\152,172,184,232}&$\emptyset$&$\emptyset$&13\\\hline
{\makecell{There is a sequential\\universal update mode}}&$\emptyset$&\makecell{2,10,24,26,34,\\42,58,130,138,\\154,162,170}&$\emptyset$&$\emptyset$&12\\\hline
{\makecell{There is a covering of\\sequential update modes}}&$\emptyset$&\makecell{50,74,104,178}&\makecell{18,122\\146}&$\emptyset$&7\\\hline
{\makecell{There is no covering of\\sequential update modes}}&$\emptyset$&$\emptyset$&90&106&2\\\hline
\end{tabular}}
\caption{Summary of the classification of the rules that fulfill at least one of the conditions stated by Theorem~\ref{thm:DAS}.}
\label{tab:inside-theorem}
\end{table}

Furthermore, we found that 7 of the 50 ECAs have a covering of different sequential update mode in order to guarantee that all configurations reach a fixed point. Contrarily, for Rule 90 we found a sufficient condition under which a configuration cannot reach a fixed point under any sequential update mode. Additionally, while we did prove that there exists an $E$-universal sequential update mode for Rule $106$, we found there does not exist a covering of sequential update modes for it. Note that this does not contradict Theorem~\ref{thm:DAS}, since said theorem refers to asynchronous update modes, which is a more general notion than sequential update modes.

In the case of the 38 ECAs that do not fall under the purview of Theorem~\ref{thm:DAS}, we can easily distinguish between the 25 ECAs that have fixed points under parallel update mode, and 13 ECAs that do not have any. Within the first group we have 2 ECAs for which we were able to prove, as well as Rule $45$, for which we have a conjecture, that there are sequential update modes that lead all configurations to a fixed point by demanding extra conditions over the size of the ring.
Additionally, we found 2 ECAs for which we need more than one sequential update mode to allow all configurations to arrive at a fixed point. 

Moreover, we found 20 ECAs for which there are configurations that can never reach a fixed point under sequential update modes, divided in three categories: 12 ECAs whose fixed points are isolated, 7 ECAs where we characterized the configurations that are unable to reach fixed points, and finally Rule $37$ for which the computer simulations we ran strongly suggest that there are configurations that can never reach a fixed point under sequential update modes.

\subsection*{Future Work}
There are other families of periodic update modes, for example: local clocks~\cite{Rios2021-phd,rios2024c} and block-parallel~\cite{Demongeot2020,Perrot2024a,Perrot2024b}, that could create ``new'' fixed points for rules that do not have them under the traditional update modes, as well as for rules whose only fixed points are isolated. Indeed, more complex update modes could help us control said rules, allowing convergence when there was none.

When we proved that there are configurations that cannot reach a fixed point under sequential update modes for Rule 90, we showed that it is necessary to update some cells two times before their neighbors are updated once. This means we would have to search for an update mode such as the ones we previously mentioned; or we would need to define a periodic update mode such as the one described for Rule 45.

Naturally, it is of interest to prove Conjecture~\ref{cnjtr:rule45}, which would improve our current results by proving that Rule 45 can lead all configurations to a fixed point using only one sequential update mode, given a ring whose size is a multiple of $3$. {Analogous situation refers to} Conjecture~\ref{cnjtr:rule37}, which would settle the question of convergence under sequential update modes for Rule $37$.

Additionally, for rules with more than one universal sequential update mode, it is of interest to find out which of them takes the least amount of time to converge, as well as how does the size of the ring affects the convergence time, {so that we can} find the sequential update mode with the most efficient convergence time. 

Along similar lines, for rules for which we did not find a universal sequential update mode, it begs the question of can we find a minimal covering,~\ie, the smallest set of sequential update modes that is a covering and whether the size of such a covering be dependent on the size of the ring. At the same time, can we develop an algorithm that gives us a minimal covering for any given rule and size of ring?

\paragraph{Acknowledgements}
This work has been partially funded by the HORIZON-MSCA-2022-SE-01 pro\-ject 101131549 ``ACANCOS'' project, the ANR-24-CE48-7504 ``ALARICE'' project, the  STIC AmSud 22-STIC-02 ``CAMA'' project, ANID FONDECYT 1250984 regular and ANID-MILENIO-NCN 2024\_103.

\bibliographystyle{plain}


\end{document}